\documentclass[onefignum,onetabnum]{siamart190516}
\usepackage{graphicx}
\usepackage{amssymb}
\usepackage[margin=1in]{geometry}
\usepackage{epstopdf}
\usepackage{wasysym}
\usepackage{ dsfont }
\usepackage{multicol}
\usepackage{enumitem}
\usepackage{cite}
\usepackage{fancyvrb}
\usepackage{tikz}
\usetikzlibrary{automata,topaths,arrows,shapes,positioning,calc,decorations.pathreplacing}
\usepackage{changepage}
\usepackage{cprotect}
\usepackage[utf8]{inputenc}
\usepackage[english]{babel}
\usepackage{mathrsfs}
\usepackage{upgreek}
\usepackage{ulem}
\usepackage{amsmath}
\usepackage{bm}

\usepackage{lineno}

\newcommand{\ghostprime}{\makebox[0pt][l]{$\smash{'}$}}

\newcommand{\DMCES}{\operatorname{DMCES}}
\newcommand{\bbDMCES}{\mathbb{DMCES}}
\newcommand{\aDMCES}{\operatorname{aDMCES}}
\newcommand{\MCIS}{\operatorname{MCIS}}
\newcommand{\bbMCIS}{\mathbb{MCIS}}

\tikzset{
    position/.style args={#1:#2 from #3}{
        at=(#3.#1), anchor=#1+180, shift=(#1:#2)
    }
}

\title{A poset metric from the directed maximum common edge subgraph\thanks{\funding{The research of RN, PCK, BC and TG was partially supported by  NSF TRIPODS+X grant DMS-1839299,  DARPA FA8750-17-C-0054 and NIH 5R01GM126555-01.}}}
\author{Robert R. Nerem\thanks{Institute for Quantum Science and Technology, University of Calgary, Alberta T2N 1N4, Canada   (\email{riley.nerem@gmail.com}) }  \and Peter Crawford-Kahrl
\thanks{Department of Mathematical Sciences,
Montana State University,
Bozeman, Montana, USA \newline (\email{peter.crawford.kahrl@gmail.com}, \email{breschine.cummins@montana.edu}, \email{gedeon@math.montana.edu}) }\and  Bree Cummins\footnotemark[3]  \and Tom{\'a}{\v s} Gedeon\footnotemark[3]}

\begin{document}
	\maketitle

\begin{abstract}
We study the directed maximum common edge subgraph problem (\textbf{DMCES}) for the class of directed graphs that are finite, weakly connected, oriented, and simple. We  use \textbf{DMCES} to define  a  metric on partially ordered sets that can be represented as weakly connected directed acyclic graphs.  While most existing metrics assume that the underlying sets of the partial order are identical, and only the relationships between elements can differ, the metric defined here allows the partially ordered sets to be different. The proof that there is a metric based on \textbf{DMCES} involves the extension of the concept of line digraphs. Although this extension can be used to compute the metric by a reduction to the maximum clique problem, it is computationally feasible only for sparse graphs. We provide an alternative techniques for computing the metric for  directed graphs that have the additional property of being transitively closed.
\end{abstract}



\section{Introduction}
Many problems in science today use the language of directed graphs~\cite{BB05} to capture  relationships between agents, or features,  of a system. In parallel, partial orders may be  used to encode transitive relationships between a finite number of objects~\cite{annoni,BrugPatil2011}.  Since a partial order can be represented as a directed acyclic graph, it is natural to explore if the well-known measures of similarity between directed graphs can give rise to measures of similarity between partial orders that are meaningful with respect to the properties of a partial order. 

There have been several approaches to measuring similarity between directed graphs. Some are based on an edit distance which is given by the minimum number of elementary  operations that are needed to transform one graph to another~\cite{messmer,Bunke}, others  are based on graph isomorphism identification~\cite{jung,cornell70}, or derived from the maximum common subgraph~\cite{MCISmetric,fernandez,RASCAL,Bahiense}.  The majority of these methods assume that the graphs have the same number of nodes; for an exception that considers weighted and directed graphs see~\cite{xu}.

In this paper we study a graph metric based on the directed maximum common edge subgraph  problem (\textbf{DMCES}) for directed graphs (digraphs).  This metric is naturally extended to a  metric on partially ordered sets (posets), denoted $(P,\leq)$. A (non-strict) partial order is a binary relation $\leq$ over a set P that is reflexive, antisymmetric, and transitive. 
A partially ordered set is often represented as a directed acyclic graph, where the $\leq$ relation translates into directed edges between nodes corresponding to the elements of $P$. 
For two partially ordered sets $(P,\leq)$ and $(P',\leq')$, most metrics assume that the underlying sets of objects are identical, $P=P'$, and only the relationships between elements can differ~\cite{Cook,Fattore,Zelinka}. Our metric, defined using  \textbf{DMCES}, measures the distance between posets where the sizes of the underlying sets can be different, $|P| \neq |P'|$.

 In addition to comparing posets with different numbers of elements, we will compare partially ordered sets that  are labeled, meaning there is a function $\ell: P \to \mathscr L$ which maps elements of the poset to elements of a set of labels $\mathscr L$. The notation $(P,\leq,\ell)$ will refer to a labeled poset with labeling function $\ell$. 
 The addition of node labels is useful since labels can capture additional information relevant to the set $P$.

The transformation of a partial order to a directed graph is straightforward.
\begin{definition}
The \emph{digraph of a partial order} $(P,\leq)$ is a directed graph $\mathscr D(P,\leq)$ with vertices $P$ and a directed edge from node $v_1 \in P$ to $v_2 \in P$ if and only if $v_1 \leq v_2$ and $v_1 \neq v_2$.
 The  digraph of a labeled poset $\mathscr D(P,\leq,\ell)$ is a node-labeled graph which inherits the labeling function $\ell$ from the poset.
\end{definition}
Two consequences of this definition are that (1) $\mathscr D(P,\leq)$ is acyclic, which arises from the antisymmetry of the partial order along with the requirement that $v_1 \neq v_2$, and (2) $\mathscr D(P,\leq)$ is transitively closed. We remark that the well-known Hasse diagram of a partial order is the transitive reduction of $\mathscr D(P,\leq)$.

In this paper we develop methods to compute for two directed, labeled graphs $G$ and $G'$,  the size of  the directed maximum common edge subgraph via a function that we denote 
\[ \DMCES(G,G'). \] 
While we postpone the precise definition of $\DMCES(G,G')$ to  Definition~\ref{def:DMCES}, we use DMCES to define a metric on the digraphs of partial orders, which we use interchangeably as a distance between posets.

\begin{definition}\label{def:metric}
	Let $(P,\leq,\ell)$ and $(P',\leq',\ell')$ be two partial orders. Denote
	\[\mbox{$G = \mathscr D (P,\leq,\ell)$ and $G' = \mathscr D (P',\leq',\ell')$}\] with edge sets $\mathcal D$ and $\mathcal D'$ respectively. The distance between $(P,\leq,\ell)$ and $(P',\leq',\ell')$ is 
	\begin{equation}
		d\big((P,\leq,\ell) , (P', \leq',\ell') \big) := d_{e}(G,G') ,
	\end{equation}
	where 
	\begin{equation}\label{eq:DMCES_metric}
	d_{e}(G,G') :=	1  - \frac{ \DMCES (G,G')}{\max(|\mathcal D|,|\mathcal D'|)},
	\end{equation}
\end{definition}
where the vertical bar notation denotes the size of a set. Our motivation for using this metric is that it is the proportion of unmatched relationships $p \leq q$ between two posets $(P,\leq,\ell)$ and $(P',\leq',\ell')$. This is a natural idea of distance in the sense of partial orders.
Although the definition of $d_{e}$ is motivated by an application to posets, it is an interesting graph metric in its own right. 

Poset metrics different than Definition~\ref{def:metric}, based on the maximum common node subgraph problem (\textbf{MCIS})~\cite{Zelinka} and the maximum common edge subgraph problem (\textbf{MCES})~\cite{Haviar}, have been studied previously for posets without labels. 
These papers do not use the language of graph theory and instead focus only on posets.  Furthermore, the emphasis is on studying the property of the metric on the set of all partially ordered sets on $n$ elements and not on developing techniques to evaluate the distance between two partial orders. The \textbf{MCES} problem  has been studied for undirected graphs~\cite{RASCAL,Bahiense} and the maximum common node subgraph problem has been studied for directed graphs~\cite{Bunke}. A heuristic algorithm for \textbf{DMCES} is given in~\cite{Larsen} which is to our knowledge the only previously investigation of \textbf{DMCES} as a computational problem.

The main objectives of this paper are to prove that~\eqref{eq:DMCES_metric} satisfies the properties  of a metric (reflexivity, symmetry, and the triangle inequality) and to provide techniques for computing it. 
 To prove that $d_e$ is a  metric, an object called the \textit{extended line digraph} is introduced, which is related to the well-known line (di)graph of a graph. The extended line digraph is used to demonstrate both that~\eqref{eq:DMCES_metric} is a metric, and that \textbf{DMCES} can be reduced to the maximum clique problem as has been done for undirected graphs in \cite{RASCAL}. An important step in this process is the formulation and proof of the Isomorphism Theorem~\ref{thm:iso_iff} for a subset of labeled digraphs that is analogous to Whitney's isomorphism theorem for undirected graphs~\cite{whitney}.
 
 Algorithms based on the extended line digraph are inefficient except for sparse graphs, and therefore we introduce special methods for dense graphs. In particular, we consider transitive closures of graphs, which occur when  $\textbf{DMCES}$ is used to compute the poset metric of Definition~\ref{def:metric}. For these graphs we determine a number of properties the directed maximum common edge subgraph must satisfy, which greatly reduces the space of subgraphs over which to search. An algorithm leveraging these results is described in Appendix~\ref{sec:alg}.



\section{Preliminaries}

In this section, we establish graph theory definitions that will be used throughout the paper. We define labeled graphs and the idea of isomorphism on labeled graphs, as well as discussing important assumptions on graph properties that are used periodically in proofs.

	\begin{definition}\label{def:typesofgraphs}

		A \emph{labeled directed graph} or \emph{labeled digraph}, $G = (V,\mathcal D, \ell_v, \ell_e)$, is a graph with nodes $V$, edges $\mathcal D$, and label functions $\ell_v$ and $\ell_e$. The \emph{directed edges} $\mathcal D \subseteq V \times V$ are a set of ordered pairs of nodes. The notation $(v_1,v_2)$ will be used for  a directed edge. The \emph{node-labeling function} $\ell_v$ maps nodes $V$ onto a label set. The \emph{edge-labeling function} $\ell_e$ maps edges $\mathcal D$ onto a (possibly different) label set. 
		A \emph{labeled undirected graph} $G= (V,\mathcal E, \ell_v, \ell_e)$ is similar, except that  the \emph{undirected edges} $\mathcal E \subseteq V \times V$ form a set of unordered pairs of nodes. The notation $\{v_1,v_2\}$ will be used for  an undirected edge. 
		When a graph can be either directed or undirected, the notation $G = (V,\mathcal F, \ell_v, \ell_e)$ with edge notation $\langle u,v \rangle \in \mathcal F$ will be employed.
		
	   \end{definition}
	   
	   When labeling functions are absent, they will be replaced by the empty set notation. For example, $G=(V,\mathcal E, \ell_v, \emptyset)$ is a node-labeled undirected graph. In this manuscript, we will be concerned only with labeled, directed graphs and unlabeled, undirected graphs. For consistent notation, unlabeled, undirected graphs will always employ empty set notation, $G=(V,\mathcal E, \emptyset, \emptyset)$. 
		
To compute the metric in Definition~\ref{def:metric}, we must find the size of the largest common subgraph of two digraphs. This requires both the notions of subgraph and graph isomorphism for labeled graphs.

	\begin{definition} \label{def:subgraph}
		Let $G = (V, \mathcal F, \ell_v,\ell_e)$ be a (possibly labeled) graph.
		\begin{enumerate}
			\item Let $U \subseteq V$ and let $W \subseteq \mathcal F$ be subsets such that $\langle u,v \rangle \in W$ implies $u,v \in U$. Then $H = (U,W,\ell_v|_U,\ell_e|_W)$ is a \emph{subgraph of G}.
		\item Let $W \subseteq \mathcal F$.  The \emph{$W$ edge-induced subgraph of $G$} is a graph $H = (U,W,\ell_v|_U,\ell_e|_W)$ with $U \subseteq V$
		such that 
		$$ U = \{v_1 \in V\mid \langle v_1,v_2 \rangle \in W \text{ or } \langle v_2,v_1 \rangle \in W\}.$$
		\item Let $U \subseteq V$. The \emph{$U$ node-induced subgraph} is a graph $H =(U,W,\ell_v|_U,\ell_e|_W)$
			with $W\subseteq \mathcal F$ such that
			\begin{align*}
			    W &= \{\langle v_1,v_2\rangle \in \mathcal F \mid v_1, v_2 \in U \}.
			\end{align*}
		\end{enumerate}
	\end{definition}

	\begin{definition}\label{def:mixed_iso}
	    Let $G= (V,\mathcal F, \ell_v,\ell_e)$ and $G'= (V',\mathcal F', \ell_v',\ell_e')$ be labeled graphs that are either both directed or both undirected. Let $U \subseteq V$.
		We say a map $\phi: U \to V'$ \emph{respects labels} if 
		\begin{enumerate}
		    \item $\ell_v(v) = \ell_v'(\phi(v))$ for all  $v \in U$, and
		    \item $\ell_e( \langle u,v \rangle) = \ell_e'( \langle\phi(u),\phi(v)\rangle)$ for all $u,v \in U$ whenever $\langle u,v \rangle \in \mathcal F$ and $ \langle\phi(u),\phi(v) \rangle \in \mathcal F'$.
		\end{enumerate}   
        	The map $\phi: V \to V'$ is an \emph{isomorphism between $G$ and $G'$} if and only if $\phi$ is a label-respecting bijection such that $\langle v_1,v_2\rangle \in \mathcal F$ if and only if 
		$\langle \phi(v_1), \phi(v_2)\rangle \in \mathcal F'$. If an isomorphism exists between $G$ and $G'$, we say they are isomorphic and write $G \cong G'$.
\end{definition}

	       Many of our proofs require assumptions on graph properties, most notably the Isomorphism Theorem~\ref{thm:iso_iff}. The assumptions we will use at various times are the following.
	       \begin{itemize}[align=left,labelsep=1ex]
	           \item[\textbf{(F)}] \textbf{Finite.}
	           \item[\textbf{(W)}] \textbf{Weakly connected.} There is an undirected path between any two nodes in a directed graph (compare to strongly connected, where the path must be directed). In an undirected graph, weak connectivity is equivalent to connectivity.
	           \item[\textbf{(S)}] \textbf{Simple.} There are no self-loops and no parallel edges from the same source to the same target. For undirected graphs, this means that there are no multi-edges. For directed graphs, this means that two edges can only appear between the same two nodes if they point in opposite directions.
	           \item[\textbf{(O)}] \textbf{Oriented.} There are no 2-cycles in a directed graph; i.e. no pairs of directed edges going in opposite directions between the same nodes.
	       \end{itemize}
           
           All graphs in this manuscript fulfill \textbf{(F)}. Many graphs will additionally be assumed to fulfill \textbf{(W)}, \textbf{(S)}, and \textbf{(O)}. This includes all digraphs of labeled partial orders. Notice that these digraphs naturally fulfill \textbf{(S)} and \textbf{(O)}, but many will not fulfill \textbf{(W)}. Therefore we restrict ourselves to the class of labeled partial orders that produce weakly connected digraphs.


\section{\textbf{DMCES}}

The directed maximum common edge subgraph (\textbf{DMCES}) optimization problem is equivalent to the problem of locating a (not necessarily unique) largest common edge-induced subgraph between two node-labeled digraphs. In this section, we introduce two definitions for this problem, one more traditional than the other, and show that they are equivalent. Both definitions are used at various points in the manuscript. The first definition given below is modified from the definition given in \cite{RASCAL} for the  maximum common edge subgraph (\textbf{MCES}) problem for undirected graphs.

\begin{definition} 
	\label{def:DMCES}
	Let $G = (V, \mathcal D, \ell_v, \emptyset)$ and $G' = (V',\mathcal D',\ell_v', \emptyset)$ be node-labeled digraphs. 
	Define   $$\epsilon:V\times V  \to \{0,1\} \text{ and }\epsilon':V'\times V' \to \{0,1\}$$  by
	\[ \epsilon(v_1,v_2) := \begin{cases} 1 & \text{if }(v_1,v_2) \in \mathcal D \\
	0 & \text{otherwise} \end{cases} \quad \text{and} \quad \epsilon'(v'_1,v'_2) := \begin{cases} 1 & \text{if }(v'_1,v'_2) \in \mathcal D'\\
	0 & \text{otherwise} \end{cases} \quad .
	\] 
	Let $U \subseteq V$ and $\phi: U \to V'$ be an injection that respects labels. We refer to the ordered pair $(U, \phi)$ as a \emph{feasible solution}, and the set of all feasible solutions (to \textbf{DMCES}) as
		$$ \bbDMCES(G,G') := \{(U, \phi) \mid (U, \phi) \text{ is a feasible solution}\}$$
	
	For any $(U, \phi) \in  \bbDMCES(G,G')$, we define the \emph{score} of the feasible solution $(U,\phi)$ to be the function
		\begin{equation}\label{eq:score}
	\mathcal P (U,\phi) := \sum_{(v_1,v_2) \in U\times U} \epsilon(v_1,v_2)\epsilon'(\phi(v_1),\phi(v_2)).
	\end{equation}
Let ${\mathscr G}$ be the set of all node-labeled digraphs. We define the function
\begin{gather*}
\DMCES: {\mathscr G} \times {\mathscr G} \to \mathbb N \\
\DMCES(G,G') := \max \{\mathcal P (U,\phi) \mid (U,\phi) \in \mathbb{DMCES}(G,G')\}.
\end{gather*} 

The \emph{directed maximum common edge-induced subgraph  problem (\textbf{DMCES})} is to calculate, for inputs $G$ and $G'$, the value of $\DMCES(G,G')$.
 We call $(U,\phi)$ with  $\mathcal P(U,\phi) = \DMCES(G,G')$ a \emph{solution to \textbf{DMCES}}.
\end{definition}
The score $\mathcal P$ is the number of edges matched under $\phi$, which means \textbf{DMCES} maximizes the number of edges that can be matched under any label-respecting injection.

There is an alternative way of formulating \textbf{DMCES} that involves isomorphic subgraphs and can be more amenable to computation. We now define the alternative directed maximum common edge subgraph problem (\textbf{aDMCES}).

\begin{definition} \label{def:aMCES}
Let $G = (V, \mathcal D, \ell_v, \emptyset)$ and $G' = (V', \mathcal D',\ell_v', \emptyset)$ be node-labeled digraphs.
A \emph{feasible solution} (to \textbf{aDMCES}) is an ordered pair $(W,W')$ where $W \subseteq \mathcal D$ and $W' \subseteq \mathcal D'$, and the edge-induced subgraphs of $W$ and $W'$ are isomorphic. We denote the set of all such feasible solutions as 
\[a\bbDMCES(G,G') := \{(W,W') \mid (W,W') \text{ is a feasible solution}\} \]
and define the function 
\begin{gather*}
\aDMCES: {\mathscr G} \times {\mathscr  G} \to \mathbb N \\
\aDMCES(G,G') := \max \{|W| \mid (W,W') \in a\bbDMCES(G,G')\}.
\end{gather*} 
The \emph{alternative DMCES problem (\textbf{aDMCES})} is to calculate, for inputs $G$ and $G'$, $\aDMCES(G,G')$. We call $(W,W')$ with $|W| =  \aDMCES(G,G')$ a \emph{solution to \textbf{aDMCES}}.

\end{definition}

\begin{theorem} \label{thm:DMCEStoaDMCES}
\textbf{DMCES} is equivalent to \textbf{aDMCES} for simple node-labeled digraphs, i.e.\  
\begin{equation}
    \aDMCES(G,G') = \DMCES(G,G')
\end{equation}
whenever $G$ and $G'$ satisfy the assumption \textbf{(S)}.
\end{theorem}

\begin{proof}
Let $G = (V, \mathcal D, \ell_v, \emptyset)$ and $G' = (V', \mathcal D',\ell_v', \emptyset)$ be node-labeled digraphs satisfying assumption \textbf{(S)}. Suppose $(U,\phi) \in \bbDMCES(G,G')$. Let 
	\begin{equation} \label{eq:W}
	W :=  \{(v_1,v_2) \in \mathcal D\mid \epsilon(v_1,v_2)\epsilon'(\phi(v_1),\phi(v_2)) =1\}.
	\end{equation}	
	and let 
\[ W' :=   \{(\phi(v_1),\phi(v_2)) \in \mathcal D' \mid \epsilon(v_1,v_2)\epsilon'(\phi(v_1),\phi(v_2)) =1\} .\]
	 Let $H$ and $H'$ be the $W$ and $W'$  edge-induced subgraphs respectively. Then $\phi$ is  an isomorphism between $H$ and $H'$. To see this, we first observe $\phi$ respects labels. Second, we consider the edges.
		\begin{align*}
		(v_1,v_2) \in W \Leftrightarrow&\ \epsilon(v_1,v_2)\epsilon'(\phi(v_1),\phi(v_2)) =1 \\
		\Rightarrow&\ \epsilon'(\phi(v_1),\phi(v_2)) =1 \\
		\Leftrightarrow&\ (\phi(v_1),\phi(v_2)) \in \mathcal D'.
	\end{align*}
	Since $\epsilon(v_1,v_2)\epsilon'(\phi(v_1),\phi(v_2)) =1$, then $(\phi(v_1),\phi(v_2)) \in W'$ as well. Setting $w_1:=\phi(v_1), w_2 := \phi(v_2)$ we have
	\begin{align*}
		(w_1,w_2) \in W' \Leftrightarrow&\ \epsilon(\phi^{-1}(w_1),\phi^{-1}(w_2))\epsilon'(w_1,w_2) =1 \\
		\Rightarrow&\ \epsilon(\phi^{-1}(w_1),\phi^{-1}(w_2)) =1 \\
		\Leftrightarrow&\ (\phi^{-1}(w_1),\phi^{-1}(w_2)) \in \mathcal D\\
		\Leftrightarrow&\ (v_1,v_2) \in \mathcal D.
	\end{align*}
	 The first and last lines above imply $(v_1,v_2) \in W$.  Putting the two arguments  together, $(v_1,v_2) \in W \Leftrightarrow (\phi(v_1),\phi(v_2)) \in W'$.
	It now follows that  $(W, W') \in a\bbDMCES(G,G')$. Furthermore, by construction of the sets $W, W'$ we have 
	\begin{equation}\label{est1}
	 \mathcal P (U,\phi) = |W|=|W'|.
	 \end{equation}
	
	Now let $W \subseteq \mathcal D, W' \subseteq \mathcal D'$ such that $(W,W') \in a\bbDMCES(G,G')$, i.e.\ $W$ and $W'$ are two sets of edges that form a feasible solution to \textbf{aDMCES}. Let $H = (U, W,\ell_v|_U,\emptyset)$ and $H' = (U',W',\ell_v'|_{U'},\emptyset)$ be the edge-induced subgraphs associated with $W$ and $W'$ respectively and let  $\psi: U \to U'$ be the isomorphism between $H$ and $H'$. Then the pair $(U,\psi)\in \bbDMCES(G,G')$.
	Since $G$ and $G'$ are simple, there is at most one edge from $v_i$ to $v_j$.
	Therefore 
	\begin{equation}\label{est2}
	\mathcal P(U,\psi) = |W|=|W'|.
	\end{equation}	
	The equations  \eqref{est1}-\eqref{est2} show that there is a solution of \textbf{DMCES} with score $s$ if and only if there is a solution to \textbf{aDMCES} with score $s$.  
\end{proof}


\section{The extended line digraph}\label{sec:extlinedigraph}

In order to prove that Definition~\ref{def:metric} defines a metric, we make use of \textit{ the extended line digraph}. As its name suggests, this is an elaboration of the notion of a line graph $L(G)$ of an unlabeled undirected graph $G=(V,\mathcal E, \emptyset,\emptyset)$.
$L(G)$ forms a dual to $G$ in the sense that edges in $G$ are converted to nodes in $L(G)$. Edges in $L(G)$ occur whenever two edges in $G$ share a node. In this section, we modify the standard idea of the line graph to capture information about the arrangement of directed edges in a digraph and to account for node labels. We use this construct to prove an isomorphism theorem similar to Whitney's isomorphism theorem, which relates isomorphisms between graphs with isomorphisms between their line graphs.

We remark briefly that there is a standard idea of a line digraph of an unlabeled directed graph $G = (V,\mathcal D,\emptyset,\emptyset)$. In the line digraph, the head-to-tail relationships between edges of $G$ become directed edges in the line digraph. In our case, we wish to capture all edge relationships, whether head-to-tail, tail-to-tail, or head-to-head, in order to prove the Isomorphism Theorem~\ref{thm:iso_iff}. Therefore we define a custom dual graph for a node-labeled digraph.

We begin with the definition of the line graph and Whitney's isomorphism thereom. We then go on to define the extended line digraph and its isomorphism theorem.

\begin{definition}
Given an unlabeled undirected graph $G = (V,\mathcal E,\emptyset,\emptyset)$, the \emph{line graph of $G$} is an unlabeled undirected graph $L(G) = (\mathcal E, \mathcal E_L,\emptyset,\emptyset)$, with nodes that correspond to edges of $G$. The edges $\mathcal E_L$ connect nodes in $L(G)$ whenever there is a shared node between two edges $e_1,e_2 \in \mathcal E$:

\[\mathcal E_L := \big\{\{e_1,e_2\} \in \mathcal E\times \mathcal E \mid e_1 \neq e_2\text{ and } e_1 = \{v_1,v_2\} , e_2 = \{v_2,v_3\} \mbox{ for some } v_1,v_2,v_3 \in V\big \}.\]
	

\end{definition}

Whitney's isomorphism theorem holds for almost all undirected graphs. However, there is one exception, the isomorphism of the line graphs between the $Y$ and $\Delta$ graphs.

\begin{definition} The $Y$ and $\Delta$ graphs (Figure \ref{fig:dY}) are defined as
	$$Y := \big (\big \{a,b,c,d\big \}, \big\{\{a,b\},\{a,c\},\{a,d\} \big\}\big) $$
	$$\Delta := \big(\big \{a,b,c\big \},  \big\{\{a,b\},\{b,c\},\{c,a\}  \big\}\big). $$
	\end{definition}
	
	Note that  these two graphs have isomorphic line graphs, which are isomorphic to the $\Delta$ graph itself.  The Whitney isomorphism theorem states that these are the only non-isomorphic graphs that have isomorphic line graphs. 


\begin{figure}[h!]
	\centering
	\begin{minipage}{.48\linewidth}
		\centering
		\begin{tikzpicture}[main node/.style={circle, draw, inner sep=1pt,minimum size=8pt, fill=white,minimum width={width("d")+6pt}}, scale=1,node distance=0.8cm]
		\tikzstyle{every loop}=[looseness=14]
		
		\node[main node] (A) at (0,0) {$a$} ;
		\node[main node] (B) at (0,-1) {$d$} ;
		\node[main node] (C) at ({cos(30)},{sin(30)}) {$c$} ;
		\node[main node] (D) at (-{cos(30)},{sin(30)}) {$b$} ;
		
		\draw[very thick,-,shorten >= 3pt,shorten <= 3pt]
		(A) edge[] (B)
		(A) edge[] (C)
		(A) edge[] (D);
		\end{tikzpicture}
	\end{minipage}
	\hfill
	\begin{minipage}{.48\linewidth}
		\centering
		\begin{tikzpicture}[main node/.style={circle, draw, inner sep=1pt,minimum size=8pt, minimum width={width("d")+6pt}},fill=white, scale=.8,node distance=0.8cm]
		\tikzstyle{every loop}=[looseness=14]
		
		\node[main node] (A) at (1,0) {$a$} ;
		\node[main node] (B) at (0,{-2*sin(60)}) {$b$} ;
		\node[main node] (C) at (2,{-2*sin(60)}) {$c$} ;
		
		\draw[very thick,-,shorten >= 3pt,shorten <= 3pt]
		(A) edge[] (B)
		(B) edge[] (C)
		(C) edge[] (A);
		\end{tikzpicture}
	\end{minipage}
	\caption{The $Y$ and $\Delta$ graphs.}\label{fig:dY}
\end{figure}
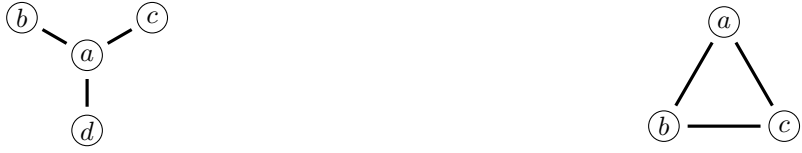

\begin{theorem}[Whitney~\cite{whitney}] 
	Let $G,G'$ be two finite, connected, unlabeled undirected graphs. Then $$L(G) \cong L(G') \text{ if  and only if } G \cong G',$$ with the  exception when $G \cong Y$ (respectively $G' \cong Y$) and $G' \cong \Delta$ (respectively $G \cong \Delta$).  
		\end{theorem}

Now we show a similar result for node-labeled digraphs using a construction that we call the extended line digraph.

\begin{definition} 
	
	Given a  node-labeled digraph $G = (V,\mathcal D,\ell_v,\emptyset)$ satisfying \textbf{(S)}, its \emph{extended line digraph} is a labeled directed graph $\mathcal L(G) = (\mathcal D, \mathcal D_L,\bar \ell_v,\bar \ell_e)$ with node set $\mathcal D$. We label each node  in $\mathcal L(G)$ by the pair of node labels associated to the corresponding edge of $G$ 
	   $$\bar \ell_v: (u,v) \mapsto (\ell_v(u),\ell_v(v)).$$ 
	   The edges $\mathcal D_L$ and their labels $\bar \ell_e$ are determined by the head-to-tail ($ht$), tail-to-tail ($tt$), and head-to-head ($hh$) relationships between edges in $\mathcal D$. Let $e = (u,v)$ and $\hat e=(\hat u, \hat v)$ be edges in $\mathcal D$.
	\begin{itemize}
	\item 
	 $(e,\hat e) \in \mathcal D_L$ with $\bar \ell_e((e,\hat e)) = ht$ if and only if $v = \hat u$, meaning head meets tail in $G$. 
	 \item 
	 $ (e,\hat e), (\hat e, e) \in \mathcal D_L$ with $\bar \ell_e((e,\hat e)) = \bar \ell_e((\hat e, e)) = tt$ if and only if  $u = \hat u,$ meaning tail meets tail in $G$.
	 \item $(e,\hat e), (\hat e, e) \in \mathcal D_L$ with $\bar \ell_e((e,\hat e)) = \bar \ell_e((\hat e,e)) = hh$ if and only if  $ v = \hat v, $  meaning head meets head in~$G$.
	 \end{itemize}
	
   \end{definition}

   An example is seen in Figure~\ref{fig:eldexample}.

   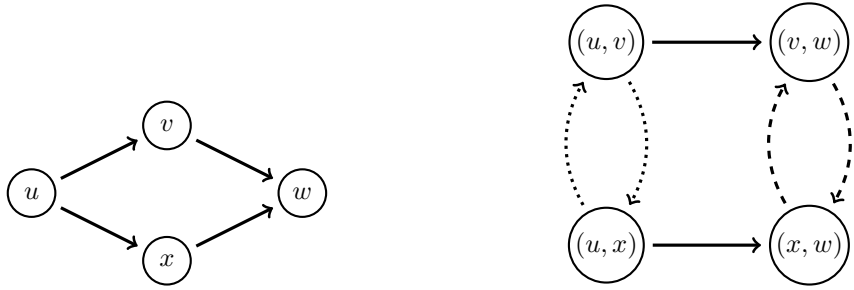
\begin{figure}[h!]\label{fig:eldexample}
	\begin{center}
		\begin{tabular}{cc}
		\begin{tikzpicture}[main node/.style={circle, draw,thick, inner sep=1pt,minimum size=18pt}, scale=.9,node distance=1cm]
		\tikzstyle{every loop}=[looseness=14]
		
		\node[main node] (u) at (0,0) {$u$ };
		\node[main node] (v) at (2,1)  {$v$ };
		\node[main node] (w) at (4, 0 )  {$w$ } ;
		\node[main node] (x) at (2,-1) {$x$};
		
		\draw[very thick,shorten >= 3pt,shorten <= 3pt]
		(u) edge[->] (v)
		(u) edge[->] (x)
		(v) edge[->] (w)
		(x) edge[->] (w)
		;
		\end{tikzpicture}
		 
		&
		\hspace{1in}
		\begin{tikzpicture}[main node/.style={circle, draw,thick, inner sep=1pt,minimum size=18pt}, scale=.9,node distance=1cm]
			\tikzstyle{every loop}=[looseness=14]
			
			\node[main node] (uv) at (0,0) {$(u,v)$ };
			\node[main node] (vw) at (3,0)  {$(v,w)$ };
			\node[main node] (ux) at (0,-3 )  {$(u,x)$ } ;
			\node[main node] (xw) at (3,-3) {$(x,w)$};
			
			\draw[very thick,shorten >= 3pt,shorten <= 3pt]
			(uv) edge[->] (vw)
			(ux) edge[->] (xw)
			(ux) edge[->,bend left,dotted] (uv)
			(uv) edge[->,bend left,dotted] (ux)
			(vw) edge[->, bend left,dashed] (xw)
			(xw) edge[->,bend left,dashed] (vw)
			;
			\end{tikzpicture}

	\end{tabular}
	\end{center}
	\caption{An example of a directed graph (left, node labels not shown) and its extended line digraph (right, node labels not shown). Solid lines indicate a head-to-tail relationship (label $ht$); dashed lines indicate head-to-head ($hh$); and dotted lines indicate tail-to-tail ($tt$).}
\end{figure}

Tracking  the head-to-head, tail-to-tail, and head-to-tail adjacencies in the extended line digraph allows us to prove an isomorphism theorem similar to the Whitney isomorphism theorem~\cite{whitney}.

\begin{theorem} \label{thm:iso_iff}
Let $G,G'$ be two node-labeled digraphs satisfying \textbf{(W)}, \textbf{(S)}, and \textbf{(O)}. Then  
\[
\mathcal L(G) \cong  \mathcal L(G') \mbox{ if and only if } G \cong G' .\]
\end{theorem}

Due to length and complexity, the proof of Isomorphism Theorem~\ref{thm:iso_iff} and relevant technical lemmas can be found in Appendix~\ref{app:thm_proof}. The importance  of Isomorphism Theorem~\ref{thm:iso_iff} is that weakly connected, simple, and oriented node-labeled digraphs, the collection of which we  call ${\mathscr G_{\mathbf {WSO}}} \subset \mathscr G$, are uniquely associated to extended line digraphs.	
As will be shown in the following section, this result allows us to use a standard metric on the extended line digraph to prove that~\eqref{eq:DMCES_metric} is a metric on ${\mathscr G}_{\mathbf {WSO}}$.

%

\section{Establishing an edge-based metric}\label{sec:metric}

A related problem to \textbf{DMCES} and \textbf{MCES} is the maximum common node-induced subgraph problem (\textbf{MCIS}). In this problem, the emphasis is on similarity between nodes in two graphs instead of similarity between edges in two graphs. Despite the change in emphasis, we can leverage a graph metric that uses \textbf{MCIS} to show that Definition~\ref{def:metric} defines a  metric. The general idea is that solving the \textbf{DMCES} problem over directed graphs is equivalent to solving the \textbf{MCIS} problem over their extended line digraphs.

\begin{definition} \label{def:MCIS}
	    
	   Let $G = (V,\mathcal D, \ell_v, \ell_e)$ and $G' = (V',\mathcal D', \ell'_v, \ell'_e)$ be labeled digraphs.
		A \emph{feasible solution} (to \textbf{MCIS}) is an ordered pair $(U,U')$ where $U \subseteq V$ and $U' \subseteq V'$ and the node-induced subgraphs $H=(U,W,\ell_v|_U,\ell_e|_W)$ and $H' = (U',W',\ell'_v|_{U'},\ell'_e|_{W'})$ are isomorphic. The set of \emph{feasible solutions} to \textbf{MCIS} is 
	\[ \bbMCIS(G,G') := \{ (U,U') \mid H \cong H'\}.\]
	Let $\bar{\mathscr G}$ be the set of all labeled digraphs. We define the function
\begin{gather*}
\MCIS: \bar{\mathscr G} \times \bar{\mathscr G} \to \mathbb N \\
\MCIS(G,G') := \max \{|U| \mid (U,U') \in \bbMCIS(G,G')\}.
	\end{gather*} 
	\raggedright
	We define the  \emph{maximum common node-induced subgraph problem (\textbf{MCIS})} to be the task of finding  $\MCIS(G,G')$ for inputs $G$ and $G'$.  We call a feasible solution $(U,U')$ such that $|U| = \MCIS(G,G')$  a \emph{solution to \textbf{MCIS}}.
\end{definition}

The following Lemma~\ref{lem:subgraphs} makes the important but technical point that the extended line digraph of an edge-induced subgraph of $G$ may  be constructed by taking the associated node-induced subgraph of $\mathcal L(G)$. Using this fact, we can then show that a solution to \textbf{MCIS} over extended line digraphs is  a solution to \textbf{aDMCES} over node-labeled digraphs.

\begin{lemma}\label{lem:subgraphs} 
	Let $G=(V,\mathcal D,\ell_v,\emptyset)$ and $G'=(V',\mathcal D',\ell_v',\emptyset) \in \mathscr G_\mathbf{WSO}$ be node-labeled digraphs satisfying \textbf{(W)}, \textbf{(S)}, and \textbf{(O)}, and let $\mathcal L(G)$ and $\mathcal L(G')$ be their extended line digraphs. Let $W\subseteq \mathcal D, W'\subseteq \mathcal D'$ be subsets of edges. Let $H$ and $H'$ be the $W$ and $W'$ edge-induced subgraphs of $G$ and $G'$, respectively. Let $J$ and $J'$ be the $W$ and $W'$ node-induced subgraphs of $\mathcal L(G)$ and $\mathcal L(G')$, respectively.
 Then 
 \begin{enumerate}
 \item[(1)]  $\mathcal L(H) = J$ and $\mathcal L(H') = J'$, and  
 \item[(2)]  $H \cong H'$ if and only if $J \cong J'$.
 \end{enumerate}
\end{lemma}

\begin{proof}

Let  $H=(U,W,\ell_v|_U,\emptyset) $ and $H' = (U',W',\ell_v'|_{U'},\emptyset)$ be the $W$ and $W'$ edge-induced subgraphs of $G$ and $G'$ respectively.
	Denote 
	\begin{align*}
		\mathcal L(G)&=(\mathcal D,\mathcal D_L,\bar \ell_v,\bar \ell_e) &
		\mathcal L(G')&=(\mathcal D',\mathcal D_L',\bar \ell'_v,\bar \ell'_e) \\
		\mathcal L(H) &= (W,\breve{\mathcal D}_L,\breve \ell_v,\breve \ell_e) &
		\mathcal L(H') &= (W',\breve{\mathcal D}'_L,\breve \ell'_v,\breve \ell'_e) \\
		J&=(W,{\bar{\mathcal D}}_L,{\bar \ell}_v|_W,{\bar \ell}_e|_{\bar{\mathcal D}_L}) &
		J'&=(W',{\bar{\mathcal D}}'_L, {\bar \ell}'_v|_{W'}, {\bar \ell}'_e |_{\bar{\mathcal D}'_L})
	\end{align*}
where $J$ and $J'$ are the $W$ and $W'$ node-induced subgraphs of the extended line digraphs $\mathcal L(G)$ and $\mathcal L(G')$, respectively.

(1) The node labeling of $\mathcal L(H)$ is the function $\breve \ell _v$ defined by 
$$(v_1,v_2) \mapsto (\ell_v(v_1),\ell_v(v_2)) \text{ for all } (v_1,v_2) \in W.$$
Since $\bar \ell_v$ is the same map for the expanded domain $(v_1,v_2) \in \mathcal D$, we have that 
$$\bar \ell_v|_W = \breve \ell_v\,. $$
Therefore the node-labeling on $W$ is the same in both $J$ and $\mathcal L(H)$. It remains to show identical edges and edge labels.

 Let $(e_1, e_2) \in \breve{\mathcal D}_L$ be an edge of $\mathcal L(H)$ with $\breve \ell_e((e_1, e_2)) = tt$. An equivalent statement is that $e_1,e_2 \in W$ share a tail-to-tail relationship in $G$. This is true if and only if $(e_1, e_2) \in {\mathcal D}_L$ with $\bar \ell_e((e_1, e_2)) = tt$ in $\mathcal L(G)$ by construction of the extended line digraph of $G$. Since $J$ is a subgraph of $\mathcal L(G)$ induced by nodes $W$, it must also be true that $(e_1, e_2) \in \bar{{\mathcal D}}_L$ and $\bar \ell_e |_{\bar{ \mathcal{D}}_L} = tt$. Similar arguments hold for edge labels $hh$ and $ht$, so that $\bar{{\mathcal D}}_L = \breve{\mathcal D}_L$ and ${\bar \ell}_e|_{\bar{ \mathcal{D}}_L} = \breve \ell_e$. Therefore $J = \mathcal L(H)$. Repeating this argument with primed objects gives $J' = \mathcal L(H')$. This shows the first statement of the Lemma.

(2) Theorem~\ref{thm:iso_iff} applied to  digraphs $H$ and $H'$ gives that $\mathcal L(H) \cong \mathcal L(H')$ if and only if $H \cong H'$, which is equivalently $H \cong H'$ if and only if $J \cong J'$. This concludes the proof.
\end{proof}

\begin{lemma}\label{lem:DMCEStoMCISsolutions}
	Consider node-labeled digraphs $G=(V,\mathcal D,\ell_v,\emptyset)$ and $G'=(V',\mathcal D',\ell_v',\emptyset)$ satisfying \textbf{(W)}, \textbf{(S)}, and \textbf{(O)} and their extended line digraphs $\mathcal L(G)$ and $\mathcal L(G')$. 
	Then  $(W,W') \in a\bbDMCES(G,G')$ if and only if  $(W,W') \in \bbMCIS((\mathcal L(G), \mathcal L(G'))$.
	Furthermore,
	\[ \aDMCES(G,G') = \MCIS(\mathcal L(G),\mathcal L(G')).  \]
\end{lemma}

\begin{proof}
Let $W \subseteq \mathcal D$ and $W' \subseteq \mathcal D'$ with  $(W,W') \in \bbMCIS(\mathcal L(G), \mathcal L(G'))$. We adopt the same notation as in Lemma \ref{lem:subgraphs}. By the definition of \textbf{MCIS}, $J$ and $J'$ are isomorphic $W$ and $W'$ node-induced subgraphs of $\mathcal L(G)$ and $\mathcal L(G')$, respectively. 
	By Lemma~\ref{lem:subgraphs}, the isomorphism between $J$ and $J'$ exists if and only if there is an isomorphism between the $W$ and $W'$ edge-induced subgraphs $H$ and $H'$ of $G$ and $G'$,  respectively. Therefore $(W,W') \in a\bbDMCES(G,G')$ if and only if $(W,W') \in \bbMCIS(\mathcal L(G), \mathcal L(G'))$.
	For both the {\bf MCIS} and  {\bf aDMCES} problems a feasible solution $(W,W')$ is a solution if there are no other feasible solutions $(T,T')$ for which $|T| > |W|$. Thus a feasible solution to {\bf MCIS} is a solution if and only if it is a solution to {\bf aDMCES}, which implies $\aDMCES(G,G') = \MCIS(\mathcal{L}(G),\mathcal{L}(G'))$.
\end{proof}

 A well-known graph distance based on \textbf{MCIS} that satisfies the properties of a metric is given in Theorem~\ref{MCISmetric}.

\begin{theorem}[Bunke and Shearer~\cite{MCISmetric}]\label{MCISmetric}
Let $\bar {\mathscr G}$ be the set of all labeled digraphs and let $G = (V,\mathcal D, \ell_v, \ell_e)$ and $G' = (V',\mathcal D', \ell_v', \ell_e')$ be any two elements of $\bar {\mathscr G}$. Then
$$d_n: \bar {\mathscr G} \times \bar {\mathscr  G} \to [0,1]$$
defined by
\begin{equation} \label{equ:mcis_metric}
d_n(G,G') = 1  - \frac{\text{MCIS}(G,G')}{\text{max}(|V|,|V'|)}
\end{equation}
is a metric on ${\mathscr G}$. 
\end{theorem}

This metric is sufficient to show that Definition~\ref{def:metric} is a metric based on DMCES$(G,G')$.

\begin{theorem}\label{thm:metric}
Let $G = (V,\mathcal D, \ell_v, \emptyset)$ and $G' = (V',\mathcal D', \ell_v', \emptyset)$ be any two elements of $\mathscr G_{\mathbf{WSO}} \subset {\mathscr G} \subset \bar{\mathscr G}$, the set of all node-labeled digraphs satisfying \textbf{(W)}, \textbf{(S)}, and \textbf{(O)}. Let 
$$d_e: \mathscr G_{\mathbf{WSO}}  \times \mathscr G_{\mathbf{WSO}}  \to [0,1] $$
$$d_e(G,G') = 1  - \frac{\text{DMCES}(G,G') }{\text{max}\big(|\mathcal D|,|\mathcal D'|\big)}.$$ 
Then $d_e$ is a metric on $\mathscr G_{\mathbf{WSO}}$.
\end{theorem}

\begin{proof}
From Theorem~\ref{thm:DMCEStoaDMCES} and Lemma~\ref{lem:DMCEStoMCISsolutions} we have that 
\[\text{DMCES}(G,G') = \text{aDMCES}(G,G')  = \text{MCIS}(\mathcal L(G),\mathcal L(G')). \] 
Moreover, the edges of $G$ and $G'$ are the nodes of $\mathcal L(G)$ and $\mathcal L(G')$.
It follows that $$d_e( G, G') = d_n(\mathcal L(G) ,\mathcal L(G')).$$ Since $G \cong G'$ if and only if $\mathcal L(G) \cong \mathcal L(G')$ by Theorem~\ref{thm:iso_iff}, $d_e$ inherits the properties of reflexivity, symmetry, and triangle inequality from $d_n$.
\end{proof}

\section{Reduction to the maximum clique finding problem}

The metric $d_e$ can be computed using an algorithm for the maximum clique finding problem. A \textit{clique} is a set of nodes that induces a complete subgraph of $G$, a \textit{maximal clique} is a clique that cannot be made larger by the addition of any node, and a \textit{maximum clique} is a maximal clique with the largest size in a graph. It is a well-known fact that \textbf{MCIS} can be reduced to the maximum clique finding problem~\cite{cliques}. We will step through the procedure as it applies to $\mathcal L(G)$ and $\mathcal L(G')$ to show that \textbf{DMCES} can also be reduced to the maximum clique finding problem.

\begin{definition}

	Consider the extended line digraphs $\mathcal L(G) = (\mathcal D, \mathcal D_L, \bar \ell_v,\bar \ell_e)$, $\mathcal L(G') = (\mathcal D', \mathcal D'_L, \bar \ell'_v,\bar \ell'_e)$. 
	Define the \emph{compatibility graph of $\mathcal L (G)$  and  $\mathcal L (G')$}  as an unlabeled undirected graph \[C(\mathcal L(G),\mathcal L(G')) = (U, \mathcal E, \emptyset,\emptyset),\] where the nodes $U$ are a collection of pairs of nodes in $\mathcal L(G)$ and $\mathcal L(G')$ with matching labels, i.e.\
	$$U = \{(n,n') \in \mathcal D \times \mathcal D'\mid  \bar \ell_v(n) = \bar \ell'_v(n')\}.$$
	 The edges $\mathcal E$ are a collection of pairs of nodes in $C(\mathcal L(G),\mathcal L(G'))$ that satisfy a label- and edge-matching condition; specifically,
 $\big \{ (n,n'), (m,m') \big\} \in \mathcal E$  if and only if either  
 \begin{itemize}
 \item $(n,m) \in \mathcal D_L \text{ and } (n',m') \in \mathcal D'_L \text{ with } \bar\ell_e((n,m)) = \bar\ell'_e((n',m')) $ or
 \item $(n,m) \not \in \mathcal D_L \text{ and } (n',m') \not \in \mathcal D'_L$.
 \end{itemize}

\end{definition}

Maximum cliques in $C(\mathcal L(G),\mathcal L(G'))$ correspond to maximum node-induced subgraphs of $\mathcal L(G)$ and $\mathcal L(G')$~\cite{cliques}, which by Definition~\ref{def:MCIS} are solutions to \textbf{MCIS} applied to $\mathcal L(G)$ and $\mathcal L(G')$.  A demonstration of how the cliques in the compatibility graph determine a subgraph isomorphism is given in Figure \ref{fig:compatibility}.

\begin{theorem}[Barrow and Burstall~\cite{cliques}]\label{thm:barrow} Let $G$ and $G'$ be labeled graphs. If $C(G,G')$ has a maximum clique of size $N$ then
$$\MCIS(G,G') = N.$$

\end{theorem}


\begin{theorem}
\textbf{DMCES} over node-labeled digraphs satisfying \textbf{(W)}, \textbf{(S)}, and \textbf{(O)} can be reduced to the maximum clique finding problem with time polynomial in the number of nodes of $G $ and $G'$.
\end{theorem}
\begin{proof}
    From Theorem~\ref{thm:barrow}, a solution to $\MCIS(\mathcal L(G),\mathcal L(G'))$ can be computed by solving the maximum clique finding problem over $C(\mathcal L(G),\mathcal L(G'))$. From Theorem~\ref{thm:DMCEStoaDMCES} and Lemma~\ref{lem:DMCEStoMCISsolutions}, $\DMCES(G,G') = \MCIS(\mathcal L(G),\mathcal L(G'))$. So \textbf{DMCES} can be solved using the maximum clique finding problem.

    Given a node-labeled digraph $G$ satisfying \textbf{(W)}, \textbf{(S)}, and \textbf{(O)}, the construction of $\mathcal L(G)$ can be done by iterating over all pairs of edges in $G$ and determining for each  pair its  adjacency type.
Since the number of pairs of edges is polynomial in number of vertices, the construction of $\mathcal L(G)$  can be done in polynomial time. Similarly, nodes in $C(\mathcal L(G),\mathcal L(G'))$ can be computed by iterating over all edges of  $\mathcal L(G)$ and $\mathcal L(G')$, and edges of $C(\mathcal L(G),\mathcal L(G'))$ can be computed by iterating over all pairs of nodes in $C(\mathcal L(G),\mathcal L(G'))$. Therefore
 $C(\mathcal L(G),\mathcal L(G'))$ can be calculated in polynomial time as well. 
\end{proof}

 We have shown that {\bf DMCES} with input $(G,G')$ can be reduced to {\bf MCIS}  with input $(\mathcal L(G), \mathcal L(G'))$  and that can be, in turn, reduced to  the maximum clique finding problem on $C(\mathcal L(G), \mathcal L(G'))$ in polynomial time.  
Many methods for solving the {\bf MCES} problem \cite{RASCAL}, which we briefly mentioned in the introduction of {\bf DMCES}, first formulate {\bf MCES} as a maximum clique finding problem and then compute the solution using well known algorithms. 
Our results show that efficient algorithms for computing the maximum clique finding problem could be leveraged to compute the  {\bf DMCES} problem. 
However, since the size of the compatibility graph and the size of the corresponding maximal clique finding problem is proportional to the product $|\mathcal D| \cdot |\mathcal D'|$ these methods are less efficient for dense graphs.  
We address this issue in the next section.

\begin{figure}[h] \label{fig:compatibility}
	\centering
		\centering
		\begin{tikzpicture}[main node/.style={circle, draw, inner sep=1pt,minimum size=22pt},fill=white, scale=.8,node distance=0.8cm]
		\tikzstyle{every loop}=[looseness=14]
		
		\node[main node, black] (A) at (1,0) {$w$} ;
		\node[main node, dashed] (B) at (0,{2*sin(60)}) {$u$} ;
		\node[main node, dashed] (C) at (2,{2*sin(60)}) {$v$} ;
		\node[] (G) at (1,2.7) {$G$} ;
		
		\draw[very thick,-,shorten >= 3pt,shorten <= 3pt]
		(A) edge[->] (B)
		(B) edge[->] (C)
		(C) edge[<-] (A);
		\end{tikzpicture}
		\qquad
		\begin{tikzpicture}[main node/.style={circle, draw, inner sep=1pt,minimum size=22pt},fill=white, scale=.8,node distance=0.8cm]
		\tikzstyle{every loop}=[looseness=14]
		
		\node[main node] (A) at (1,0) {$w'$} ;
		\node[main node, dashed] (B) at (0,{2*sin(60)}) {$u'$} ;
		\node[main node, dashed] (C) at (2,{2*sin(60)}) {$v'$} ;
		\node[] (G) at (1,2.7) {$G'$} ;
		
		\draw[very thick,-,shorten >= 3pt,shorten <= 3pt]
		(A) edge[->] (B)
		(B) edge[->] (C)
		(C) edge[<-] (A);
		\end{tikzpicture}
		
				\begin{tikzpicture}[main node/.style={circle, draw, inner sep=1pt,minimum size=19pt},fill=white, scale=1.5,node distance=0.8cm]
		\tikzstyle{every loop}=[looseness=14]
		
		\node[main node] (A) at (1,0) {$(w,w')$} ;
		\node[main node, dashed] (B) at (0,{2*sin(60)}) {$(u,u')$} ;
		\node[main node, dashed] (a) at (4,{2*sin(60)}) {$(u,v')$} ;
		\node[main node, dashed] (b) at (-2,{2*sin(60)}) {$(v,u')$} ;
		\node[main node, dashed] (C) at (2,{2*sin(60)}) {$(v,v')$} ;
		\node[] (G) at (1,2.5) {$C(G,G')$} ;
		
		\draw[very thick,-,shorten >= 3pt,shorten <= 3pt]
		(A) edge[] (B)
		(B) edge[] (C)
		(C) edge[] (A)
		(A) edge[] (a)
		(A) edge[] (b);
		\end{tikzpicture}
		\caption{Two graphs $G$ and $G'$ along with their compatibility graph. The size three clique in the compatibility graph gives an isomorphism $\phi$ between $G$ and $G'$ where $\phi(u) = u'$,$\phi(v) = v'$ and $\phi(w) = w'$. Nodes of the same border pattern share the same label. }\label{}
\end{figure}
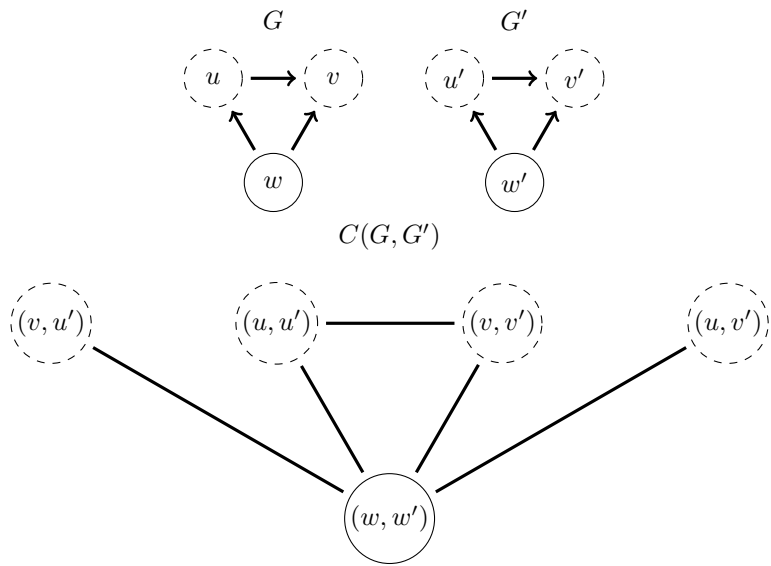



\section{Maximal cardinality solutions}

Throughout this and the following section we will use the definition of \textbf{DMCES} given by Definition \ref{def:DMCES}. Recall that a feasible solution to \textbf{DMCES} for two digraphs $G = (V, \mathcal D, \ell_v,\emptyset)$ and $G' = (V',\mathcal D',\ell'_v,\emptyset)$ is an ordered pair $(U,\phi)$, where $U \subseteq V$ and $\phi: U \to V'$ is injective and respects labels, and the set of such feasible solutions is $\bbDMCES(G,G')$. The score of each feasible solution is given by $\mathcal P (U,\phi)$, as in Equation~\eqref{eq:score}. A solution to \textbf{DMCES} is some $(U,\phi) \in \bbDMCES(G,G')$ such that $\mathcal P (U,\phi)$ is maximized, i.e. $\mathcal P (U,\phi) = \DMCES(G,G')$.

\begin{definition}
	We say a feasible solution $(U,\phi) \in \bbDMCES(G,G')$ is a \emph{maximal cardinality solution (to \textbf{DMCES})} if $\mathcal P (U,\phi) = \DMCES(G,G')$, and for all $(T,\psi) \in \bbDMCES(G,G')$, $|T| \leq |U|$.
	\end{definition}

\begin{theorem}\label{thr:size}

Let $G = (V, \mathcal D, \ell_v,\emptyset)$ and $G' = (V', \mathcal D', \ell'_v,\emptyset)$ be node-labeled digraphs satisfying \textbf{(W)}, \textbf{(S)}, and \textbf{(O)}. Then there exists a maximal cardinality solution to \textbf{DMCES}.
\end{theorem}

\begin{proof}
	First we determine the maximal  value of $|U|$ for any $(U,\phi) \in \bbDMCES(G,G')$. Let $a$ be a node label, and let $U \subseteq V$. Define $U_a := \{v \in U \mid \ell_v(v) = a\}$, i.e.\ all nodes in $U$ which have label $a$. Now define
\[ N_a(G,G') :=  \min\left\{ |\ell_v^{-1}(a)|, |\ell_v'^{-1}(a)| \right\} \]
We claim for all $(U,\phi) \in \bbDMCES(G,G')$, $|U_a| \leq  N_a(G,G')$. To see this, note $U_a = \ell_v^{-1}(a) \cap U$, so clearly $|U_a| \leq |\ell_v^{-1}(a)|$. Also, $\phi$ is an injection which respects labels, so 
\[|U_a| = |\phi(U_a)| = |\phi(U) \cap \ell_v'^{-1}(a)| \leq |\ell_v'^{-1}(a)| \]
implying 
\[ |U_a| \leq \min\left\{ |\ell_v^{-1}(a)|, |\ell_v'^{-1}(a)| \right\} =  N_a(G,G'). \]
  
  To continue the main argument we  observe that, as $U$ is a disjoint union of $U_a$,
\[ U = \bigcup_{a \in \ell_v(U)}U_a \quad \Rightarrow \quad |U| = \sum_{a \in \ell_v(U)} |U_a|. \]
We use this to obtain a bound on $|U|$, denoted $N(G,G')$, given by 
\begin{align*}
 |U| &= \sum_{a \in \ell_v(U)} |U_a| \\
 &\leq  \sum_{a \in \ell_v(U)}  N_a(G,G') \\
 &\leq  \sum_{a \in \ell_v(V)}  N_a(G,G') \\
 & =: N(G,G')
 \end{align*}

We observe that this bound holds for any feasible solution. Then, for all $ (U,\phi) \in \bbDMCES(G,G')$,  $|U| \leq N(G,G') $. Next we prove the following claim:
\[ \exists (U,\phi) \in \bbDMCES(G,G') \text{ such that } \mathcal P (U,\phi) = \DMCES(G,G') \text{ and } |U| = N(G,G')  .\]

Let $(\bar U, \bar \phi) \in \bbDMCES(G,G')$ be any feasible solution such that $\mathcal P (\bar U,\bar \phi) = \DMCES(G,G')$. 
Suppose $|\bar U|< N(G,G')$. We construct a feasible solution with the desired properties as follows. Define a $U \supset \bar U$ such that for each label $a \in \ell_v(V)$, $U$ contains $N_a(G,G')$ vertices with label $a$. Such a $U$ exists from the definition of $N_a(G,G')$. 
We first observe that $| U|=N(G,G')$. We  extend $\bar \phi$ to $\phi$ in such a way that  the restriction $\phi|_{\bar U} = \bar \phi$ and $ \phi$ is an injection which respects labels. 
This extension is possible, because the definition of $N(G,G')$ and our construction ensures that for each $a \in \ell_v(V)$, $|U_a| \leq |\ell_v'^{-1}(a)|$, so there is an injection $U_a \to V'$ that respects labels. We can then assemble these injections piecewise.

Finally, note that 
\[ \mathcal P(U,  \phi) \geq \mathcal P(\bar U,\bar \phi) \]
because $\bar U \subset  U$, $\phi|_{\bar U} =  \bar \phi$, and $\mathcal P$ is a sum of nonnegative terms, see Equation~\eqref{eq:score}. Since $\mathcal P (\bar U, \bar \phi) = \DMCES(G,G')$ it follows that 
$ \mathcal P( U,  \phi) = \mathcal P( \bar U, \bar \phi)$. Therefore $( U,  \phi)$ is the solution advertised in the Theorem.
\end{proof}
A simple example of how Theorem \ref{thr:size} applies is given in Figure \ref{fig:Cardinality example}. 
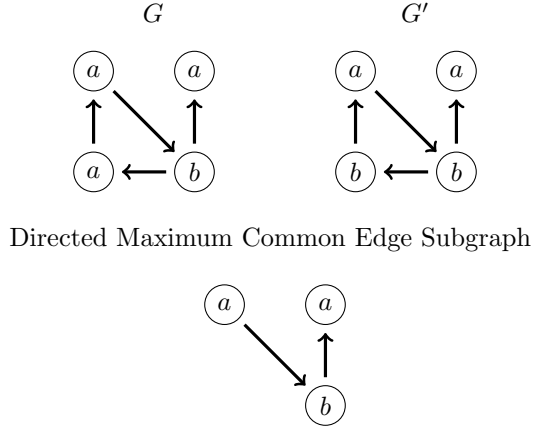
\begin{figure}[H] \label{fig:Cardinality example}
\begin{center}
\begin{tikzpicture}[main node/.style={circle, draw, inner sep=1pt,minimum size=15pt}, scale=1,node distance=.8cm]

\node [main node] (n1) {$a$};
\node [main node, right=of n1] (n2) {$a$};
\node [main node, below=of n1] (n3) {$a$};
\node [main node, right=of n3] (n4) {$b$};

\node [main node, position=0:{1.6cm} from n2] (n5) {$a$};
\node [main node, right=of n5] (n6) {$a$};
\node [main node, below=of n5] (n7) {$b$};
\node [main node, below=of n6] (n8) {$b$};

\node [position=0:{.4cm} from n1] (g) {};
\node [position=90:{.4cm} from g] (G) {$G$};

\node [position=0:{.4cm} from n5] (g') {};
\node [position=90:{.4cm} from g'] (G) {$G'$};

\draw[very thick, ->,shorten >= 3pt,shorten <= 3pt] (n1) -- (n4);
\draw[very thick, ->,shorten >= 3pt,shorten <= 3pt] (n3) -- (n1);
\draw[very thick, ->,shorten >= 3pt,shorten <= 3pt] (n4) -- (n3);
\draw[very thick, ->,shorten >= 3pt,shorten <= 3pt] (n4) -- (n2);

\draw[very thick, ->,shorten >= 3pt,shorten <= 3pt] (n5) -- (n8);
\draw[very thick, ->,shorten >= 3pt,shorten <= 3pt] (n7) -- (n5);
\draw[very thick, ->,shorten >= 3pt,shorten <= 3pt] (n8) -- (n7);
\draw[very thick, ->,shorten >= 3pt,shorten <= 3pt] (n8) -- (n6);
\end{tikzpicture} 

\vspace{.3cm}
Directed Maximum Common Edge Subgraph
\vspace{.4cm}

\begin{tikzpicture}[main node/.style={circle, draw, inner sep=1pt,minimum size=15pt}, scale=1,node distance=.8cm]

\node [main node] (n1) {$a$};
\node [main node, right=of n1] (n2) {$a$};
\node [main node, below=of n2] (n4) {$b$};

\draw[very thick, ->,shorten >= 3pt,shorten <= 3pt] (n1) -- (n4);
\draw[very thick, ->,shorten >= 3pt,shorten <= 3pt] (n4) -- (n2);
\end{tikzpicture}
\end{center}
\caption{Shown are two graphs $G$ and $G'$ along with their directed maximum common edge subgraph. Here, letters represent node labeling. Since $G$ has three $a$-labeled nodes and $G'$ has two $a$-labeled nodes, the directed maximum common edge subgraph can only have $N_a(G,G') = \min(2,3) = 2$ $a$-labeled nodes.  }\label{}
\end{figure}

\section{The order-respecting property}

This section establishes a technical property that can be leveraged in algorithms for calculating the graph distance metric established in Section~\ref{sec:metric} for graphs that are transitively closed.

\begin{definition} Let $G = (V, \mathcal D, \ell_v,\emptyset)$ and $G' = (V', \mathcal D', \ell'_v,\emptyset)$ be node-labeled digraphs. We say a feasible solution $(U, \phi) \in \bbDMCES(G,G')$ \emph{respects order on labels} if there is no $v,u \in U$  with $\ell(v) = \ell(u)$ such that  $(v,u) \in \mathcal D$ and $(\phi(u),\phi(v)) \in \mathcal D'$. 
\end{definition}

We will show that if $G$, $G'$ are transitive closures and $(U,\phi)$ is a solution to \textbf{DMCES}, then $(U,\phi)$ has the order-respecting property. We do this by taking a feasible solution $(U,\phi)$ that is not order-respecting, ``untwist'' one of the non-order-respecting edges into a map $\psi$, and show that a higher score $\mathcal P(U,\psi)$ results.

\begin{definition}
Let $G = (V, \mathcal D, \ell_v,\emptyset)$ and $G' = (V', \mathcal D', \ell'_v,\emptyset)$ be node-labeled digraphs.
For all $(U,\phi) \in \bbDMCES(G,G')$, define 
\[
\mathcal X(U,\phi) := \Big\{ \{v_1,v_2\} \subseteq U \mid \ell(v_1) = \ell(v_2), (v_1,v_2) \in \mathcal D \text{ and } (\phi(v_2),\phi(v_1)) \in \mathcal D' \Big\} .
\] 
Let $(U,\phi) \in \bbDMCES(G,G')$ such that $\mathcal X(U,\phi) \neq \emptyset$, and take an element $\{u,v\} \in \mathcal X(U,\phi)$. Let $\psi: U \to V'$ be an injection which is identical to $\phi$, with the exception that $\psi(u) = \phi(v)$ and $\psi(v) = \phi(u)$. We say that $\psi$ is a \emph{minimally untwisted map of $\phi$}. 
\end{definition}

The following  Lemma \ref{lem:alwaysIncrease} is a generalization of Lemma 2 of \cite{Haviar}. 

\begin{lemma}\label{lem:alwaysIncrease}
Let $G = (V, \mathcal D, \ell_v,\emptyset)$ and $G' = (V', \mathcal D', \ell'_v,\emptyset)$ be node-labeled digraphs satisfying \textbf{(W)}, \textbf{(S)}, and \textbf{(O)} that are transitive closures. 
Let $(U,\phi) \in \bbDMCES(G,G')$ be such that $\phi$ has a minimally untwisted map $\psi$.
Then 
\[
\mathcal P(U,\phi)  < \mathcal P(U,\psi). 
\]
\end{lemma}

\begin{proof}

Recall the definition of $\mathcal P$, given in Equation~\eqref{eq:score}, 
\begin{equation*}
	\mathcal P (U,\phi) := \sum_{(v_1,v_2) \in U\times U} \epsilon(v_1,v_2)\epsilon'(\phi(v_1),\phi(v_2)).
\end{equation*}

Let $\{u,v\} \in \mathcal X(U,\phi)$ and set
\[
C(u,v) := \sum_{\substack{(v_1,v_2) \in U\times U \\ v_1,v_2 \notin \{u,v\}}} \epsilon(v_1,v_2)\epsilon'(\phi(v_1),\phi(v_2))
\]
For each $x \in U\setminus \{u,v\}$, let $U_{(x)} := \{u,v,x\}$. The proof of this Lemma relies on the observation that
\begin{equation} \label{formula}
\mathcal P(U,\phi) = \left(\sum_{x\in U \setminus \{u,v\}} \left( \mathcal P(U_{(x)}, \phi|_{U_{(x)}})\right) \right)+ C(u,v)
\end{equation}
To see this, notice that for a fixed $x$, we have
\begin{multline*}
\mathcal P(U_{(x)}, \phi|_{U_{(x)}}) := 
\epsilon(x,u)\epsilon'(\phi(x),\phi(u)) 
+ \epsilon(x,v)\epsilon'(\phi(x),\phi(v)) 
+ \epsilon(u,x)\epsilon'(\phi(u),\phi(x)) \\
+ \epsilon(u,v)\epsilon'(\phi(u),\phi(v))
+ \epsilon(v,x)\epsilon'(\phi(v),\phi(u))
+ \epsilon(v,u)\epsilon'(\phi(v),\phi(u)).
\end{multline*}
So for each $x\in U\setminus \{u,v\}$, we have the term $\epsilon(u,v)\epsilon'(\phi(u),\phi(v))+\epsilon(v,u)\epsilon'(\phi(v),\phi(u))$, which will be multiply counted.
However, since $G$ and $G'$ are simple and oriented  and we assume that $\{u,v\} \in \mathcal X(U,\phi)$, this means that the orientation of the edge $(u,v) \in \mathcal D$ (resp. $(v,u) \in \mathcal D$) and $(\phi(u), \phi( v)) \in \mathcal D'$ (resp. $(\phi(v), \phi(u)) \in \mathcal D'$) do not agree. By the definition of $\mathcal P$, the term must be zero. This verifies the formula~\ref{formula}.


We now  observe that for the  new function $\psi$ 
\begin{equation}
\mathcal P(U,\psi) = \left(\sum_{x\in U \setminus \{u,v\}} \left( \mathcal P(U_{(x)}, \psi|_{U_{(x)}}) -1 \right) \right)+ C(u,v) + 1
\end{equation}
To explain this formula, observe that exactly one of summands
\[ \epsilon(v,u)\epsilon'(\phi(v),\phi(u)) \quad \text{} \quad  \epsilon(u,v)\epsilon'(\phi(u),\phi(v)) \]
will be equal to $1$, by the assumption that $G$ and $G'$ are oriented, while the other will be zero. To avoid counting this term multiple times we subtract it from the sum over all subgraphs indexed by $x$ and add  $+1$ at the end of the equation to account for this term exactly  once.

We will  now show that 
\begin{equation}\label{eq:alwaysIncrease}
\mathcal P(U_{(x)}, \phi|_{U_{(x)}}) + 1\leq \mathcal P(U_{(x)}, \psi|_{U_{(x)}})
\end{equation}
which will be sufficient to complete the proof. Shown in Appendix~\ref{app:figs} are all possible arrangements of the subgraphs induced by $U_{(x)} = \{u,v,x\}$, under the assumptions that $G$ and $G'$ are weakly connected, simple, oriented, and transitively closed. We use the notation $x' = \phi(x) = \psi(x) $, $v' = \psi(v) $, $u' = \psi(u)$, $v' = \phi(u) $, and $u' = \phi(v)$. In each case, the  Equation~\eqref{eq:alwaysIncrease} is valid.
\end{proof}

We remark that without the assumption of transitive closure, we cannot guarantee that $\mathcal P(U,\phi) < \mathcal P(U,\psi) $. For example,  Figure~\ref{fig:counterexampleUntwist} shows an instance where $\mathcal P(U,\phi) = \mathcal P(U,\psi)$, when $G$ and $G'$ are not transitive closures.

\begin{figure}[htp!]
\begin{center}
	\begin{tikzpicture}[main node/.style={circle, draw,thick, inner sep=1pt,minimum size=18pt}, scale=1,node distance=1cm]
	\tikzstyle{every loop}=[looseness=14]
	
	\node[main node] (v) at (0,0) {$v$} ;
	\node[main node] (u) at (2,0) {$u$} ;
	\node[main node] (x) at (1, {2*cos(30)} ) {$x$} ;
	
	\draw[very thick,-,shorten >= 3pt,shorten <= 3pt]
	(v) edge[black, ->] (u)
	(x) edge[dashed, <-] (v)
	;
	
	\node[main node] (v') at (4,0) {$v\ghostprime$} ;
	\node[main node] (u') at (6,0) {$u\ghostprime$} ;
	\node[main node] (x') at (5, {2*cos(30)} ) {$x\ghostprime$} ;
	\node[] (text') at (3,-.75) {$\mathcal P(U_{(x)},\phi|_{U_{(x)}})  = 1 $};
	\node[] (text') at (3,-1.25) {$\mathcal P(U_{(x)},\psi|_{U_{(x)}})  = 1 $};
	
	\draw[very thick,-,shorten >= 3pt,shorten <= 3pt]
	(v') edge[black, ->] (u')
	(x') edge[dashed, <-] (u')
	;
	\end{tikzpicture}.
\end{center}
\caption{An example showing that without the assumption that $G$ and $G'$ are transitive closures, we cannot guarantee that $\mathcal P(U_{(x)}, \phi|_{U_{(x)}}) + 1\leq \mathcal P(U_{(x)}, \psi|_{U_{(x)}})$. In this case, the score remains unchanged under the  new map $\psi$.  Here  $x' = \phi(x) = \psi(x) $, $v' = \psi(v) $, $u' = \psi(u)$ and $v' = \phi(u) $, $u' = \phi(v)$. The solid edges are matched in 
$(U_{(x)},\psi|_{U_{(x)}})$, and the dashed edges are matched in $(U_{(x)},\phi|_{U_{(x)}})$.}  \label{fig:counterexampleUntwist}
\end{figure}
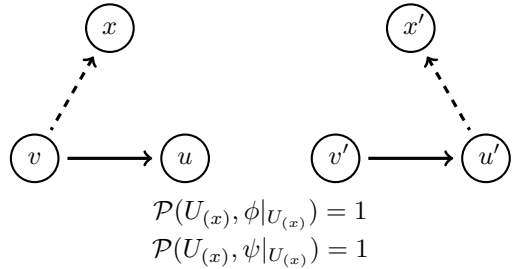

\begin{theorem} \label{thm:solutionsAreUntwisted}
Let $G = (V, \mathcal D, \ell_v,\emptyset)$ and $G' = (V', \mathcal D', \ell'_v,\emptyset)$ be node-labeled digraphs satisfying \textbf{(W)}, \textbf{(S)}, and \textbf{(O)} that are transitive closures.  If $(U,\phi)\in \bbDMCES(G,G')$ such that $\mathcal P (U,\phi) = \DMCES(G,G')$, then $(U,\phi)$ respects order on labels.
\end{theorem}

\begin{proof}
Suppose not. Then there exists $(U,\phi)\in \bbDMCES(G,G')$ with $\mathcal P (U,\phi) = \DMCES (G,G')$ and $\mathcal X(U,\phi) \neq \emptyset$. Let $\psi: U \to V'$ be a minimally untwisted map of $\phi$. Then by Lemma~\ref{lem:alwaysIncrease}, $\mathcal P(U,\phi) < \mathcal P(U,\psi)$, which is a contradiction, since we assumed $\mathcal P (U,\phi)$ was maximal. 
\end{proof}

\begin{corollary}\label{thr:direction}

Let $G = (V, \mathcal D, \ell_v,\emptyset)$ and $G' = (V', \mathcal D', \ell'_v,\emptyset)$ be node-labeled digraphs satisfying \textbf{(W)}, \textbf{(S)}, and \textbf{(O)} that are transitive closures. Then there exists $(U,\phi) \in \bbDMCES(G,G')$ such that  both 
\begin{itemize}
\item $(U,\phi)$ is a maximal cardinality solution to \textbf{DMCES} and 
\item $(U,\phi)$ respects order on labels. 
\end{itemize}
\end{corollary}

\begin{proof}
Immediate from Theorem~\ref{thr:size} and Theorem~\ref{thm:solutionsAreUntwisted}.
\end{proof}

Corollary \ref{thr:direction} greatly reduces the number of feasible solutions which must be checked by an algorithm performing an exhaustive search.  Such an algorithm taking advantage of these results is given in Appendix~\ref{sec:alg}. 

\section{Discussion}
We have constructed  a graph distance metric $d_e$ in Definition~\ref{def:metric}
that is based on an edge-induced maximum common subgraph for node-labeled digraphs. To prove that we have defined a metric, we introduce a modified version of the well-known line graph, called the extended line digraph, which captures edge direction and node labeling of a digraph $G$. We then establish the Isomorphism Theorem~\ref{thm:iso_iff} that states that weakly connected, simple, and oriented node-labeled digraphs are isomorphic if and only if their extended line digraphs are isomorphic. This allows us to compute node-induced subgraphs of the extended line digraph, $\mathcal L(G)$, instead of edge-induced subgraphs of $G$. A metric using node-induced subgraphs on $\mathcal L(G)$ then induces the metric $d_e$ on $G$. 

We further show that finding a maximum common node-induced subgraph of $\mathcal L(G)$ and $\mathcal L(G')$ can be efficiently reduced to the maximum clique finding problem. Although this method is effective for sparse graphs, it is prohibitively expensive for dense ones. Since our interest is in transitively closed graphs induced by partial orders, a different algorithm is necessary. 
To construct such an algorithm, we first show that every \textbf{DMCES} must admit a maximum cardinality solution, which is a solution with the maximum possible number of nodes with matching labels. We then prove for transitive closures that there is an order-preserving property that directed maximum common edge subgraphs must satisfy. These two properties together provide substantial savings in computational time for algorithms that exhaustively search feasible solutions. 

The need for a graph distance  metric based on the directed maximum common edge  subgraph  was motivated by a problem to assess similarity between partially ordered sets arising from biological data. 
The quantity 
$$ \frac{ \DMCES (G,G')}{\max(|\mathcal D|,|\mathcal D'|)} = 1 - d_e$$
is the maximal proportion of shared edges between two digraphs of partial orders, $G = \mathscr D(P,\leq,\ell)$ and $G' = \mathscr D(P',\leq',\ell')$, over all subgraph isomorphisms $\phi$. Therefore, it describes the maximal proportion of shared order relationships $p \leq q$ and $\phi(p) \leq \phi(q)$ between $(P,\leq,\ell)$ and $(P',\leq',\ell')$. Subtracting this proportion from 1 is then a natural measure of distance between the partial orders.
We successfully used the third algorithm in Appendix~\ref{sec:alg} to assess the similarity between collections of time series gene expression data in~\cite{tspo,sciencepaper}.


In~\cite{tspo}, we introduce a technique for representing a time series dataset as a partial order on the extrema of the time series as a function of noise. In other words, we characterize a time series dataset as a collection of peaks and troughs in gene expression that are either unambiguously ordered with respect to each other or are incomparable. The number of extrema and their ordering relationships vary assuming different levels of noise in the dataset.
Using this technique, we assess the similarity of RNAseq time series for yeast cell cycle genes in two replicate experiments. Similarity is defined to be the complement of distance, $1-d_e$, calculated between digraphs of the partial orders generated from the extrema of the time series.
By subsampling from the collection of cell cycle genes, we were able to show a high degree of similarity between the replicates, allowing the experimentalists to quantify the replicability of the experiment.

In~\cite{sciencepaper}, we used the graph distance metric $d_e$  to assess the conservation of gene peak and trough ordering in malarial parasite intrinsic oscillations compared to that in circadian rhythm datasets. We examined the similarity of peak and trough ordering across different strains of malaria parasites and different mouse tissues. Since there remain computational issues with large partial orders, we took 5000 random samples of six genes from a list of over 100 oscillating genes in malarial RNAseq data and computed the associated partial orders as in~\cite{tspo}. Given the number of extrema in each time series, the choice of six genes resulted in digraphs of partial orders of about 40 nodes. We computed the digraph distances from each sample to a reference dataset. 
A similar procedure was performed across all mouse tissues. Since it was  not \textit{a priori} clear what constitutes a large distance versus a small one, we randomized the data to create a baseline, which provided us with a null hypothesis for testing. We showed that the malaria strains conserve gene ordering more robustly 
with regard to baseline than circadian genes do across mouse tissues. Since the ordering of circadian genes are accepted to be conserved, we draw the conclusion that the gene ordering in malaria parasites is at least as conserved.

As these two examples illustrate, the graph distance metric given in Definition~\ref{def:metric} has practical application to biological datasets, as well as being an alternative metric for weakly connected, simple, and oriented node-labeled digraphs.

 \bibliographystyle{siamplain}
\bibliography{mces}

\section*{Acknowledgements}
We thank Dr. Sean Yaw for helpful discussions at an early stage of the manuscript.


\appendix

\section{Proof of Theorem~\ref{thm:iso_iff}}\label{app:thm_proof}

We define a weak notion of isomorphism that will be used extensively in the proof of Theorem~\ref{thm:iso_iff} to prove the stronger isomorphism property in Definition~\ref{def:mixed_iso}. The following definition holds for either directed or undirected graphs.

\begin{definition}\label{def:structure}
	Let $ G = (V, \mathcal F, \ell_v, \ell_e)$ be a (possibly labeled) graph. 
	The \emph{structure of $G$} is the undirected unlabeled graph
	\[S(G) = (V, \mathcal E_S, \emptyset, \emptyset) \text{ with } \mathcal E_S = \{ \{u,v\} \;|\; \langle u,v \rangle \in \mathcal F \}.\]     
	We say that two graphs $G = (V,\mathcal F,\ell_v,\ell_e)$ and $G' = (V',\mathcal F',\ell'_v,\ell'_e)$ are \emph{structurally isomorphic}  if $S(G) \cong S(G')$. For an undirected graph $G = (V,\mathcal E,\ell_v,\ell_e)$, $S(G)$ has edges $\mathcal E_S = \mathcal E$.
	\end{definition}

	\begin{lemma} \label{lem:flip} Let $G = (V,\mathcal D,\ell_v,\emptyset)$ be a node-labeled digraph satisfying assumptions \textbf{(S)} and \textbf{(O)}. Then the  structure of the extended line digraph is isomorphic to the line graph of the structure of $G$,  
\[ S(\mathcal L(G)) \cong L(S(G)).  \] 
\end{lemma}

\begin{proof}
	 We have that
  $S(G) = (V,\mathcal E_S,\emptyset,\emptyset)$ where  $\mathcal E_S = \{ \{v_1,v_2\} \;|\;(v_1,v_2) \in \mathcal D\}$, so that $$L(S(G)) = (\mathcal E_S, \mathcal E_L,\emptyset,\emptyset),$$ where $\{e_1,e_2\} \in \mathcal E_L$ if and only if $e_1$ and $e_2$ share a node. 
On the other hand, 
\[\mathcal L(G) = (\mathcal D,{\mathcal D}_L, \bar \ell_v, \bar \ell_e).\] The structure of the extended line digraph is then
\[S(\mathcal L (G)) = (\mathcal D,\bar{\mathcal E}_S,\emptyset,\emptyset),\] 
where $\bar {\mathcal E}_S$ is the set of unordered pairs of directed edges in $\mathcal{D}$ corresponding to edges in $\mathcal D_L$, each of which indicates a head-to-head, tail-to-tail, or head-to-tail relationship. 
We desire to show that $(\mathcal E_S, \mathcal E_L,\emptyset,\emptyset) \cong (\mathcal D,\bar{\mathcal E}_S,\emptyset,\emptyset)$.

Since $G$ is a simple and oriented digraph, there is  a bijection  
$\phi : \mathcal{D} \to \mathcal E_S$ defined by
$$\phi:(v_1,v_2) \mapsto \{v_1,v_2\}.$$
Suppose  $e_1 = (u,v) \in \mathcal D $ and $e_2 = (w,z) \in \mathcal D$. Then $\{\phi(e_1),\phi(e_2)\} \in \mathcal E_L$  if and only if  $\{u,v\} \cap \{w,z\} \neq \emptyset$. Therefore 
$e_1$ and $e_2$ share a head-to-tail, head-to-head, or tail-to-tail connection. This is true if and only if 
\[\{e_1,e_2\} \in \bar{\mathcal E}_S.\] 
Thus $\phi$ is an isomorphism between $S(\mathcal L(G))$ and $L(S(G))$.
\end{proof}

	\begin{lemma} \label{lem:delta_y}
	Let $G$ and $G'$ be two node-labeled digraphs satisfying \textbf{(W)}, \textbf{(S)}, and \textbf{(O)} along with the additional conditions  $S(G) \cong \Delta$ and $S(G') \cong Y$. 
	Then $\mathcal L(G) \not \cong \mathcal L(G')$.
	\end{lemma}
	\begin{proof}
	In Figure~\ref{fig:delta_y} we construct  the extended line digraphs of all weakly connected, simple, and oriented digraphs structurally isomorphic to $\Delta$ or $Y$ (up to node relabeling).  Since no two graphs in the right column of Figure~\ref{fig:delta_y} are isomorphic, this proves the Lemma.
	\end{proof}

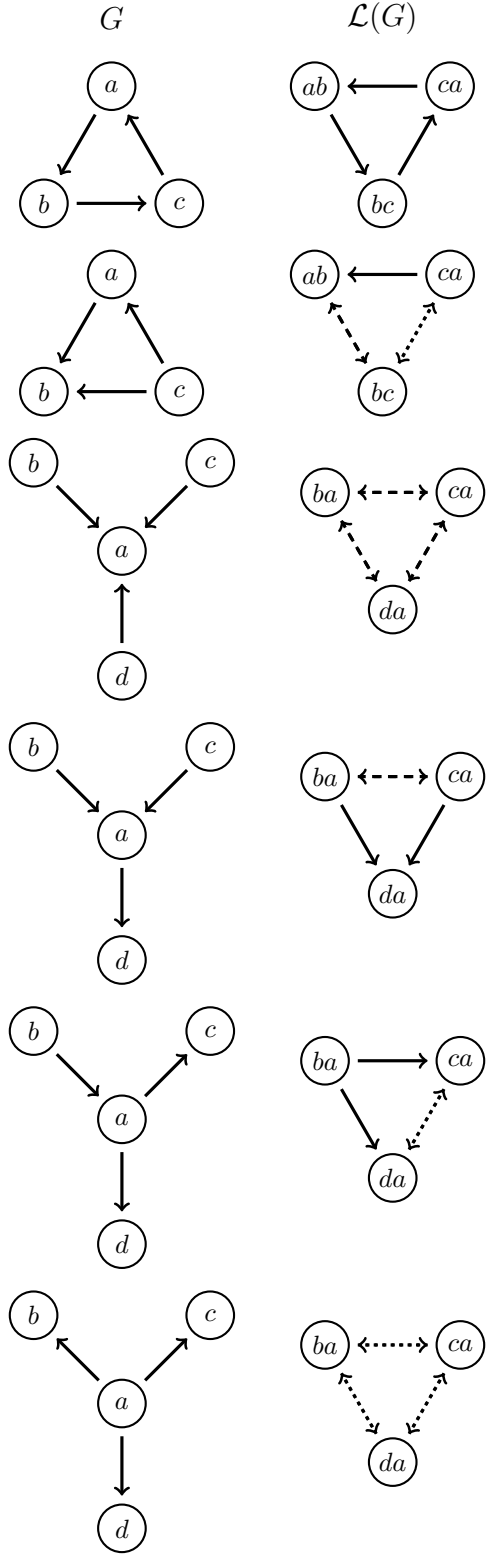
\begin{figure}\label{fig:delta_y}
	\begin{center}
		\begin{tikzpicture}[main node/.style={circle, draw,thick, inner sep=1pt,minimum size=18pt}, scale=.9,node distance=1cm]
		\tikzstyle{every loop}=[looseness=14]
		
		\node[main node] (b) at (0,0) {$b$ };
		\node[main node] (c) at (2,0)  {$c$ };
		\node[main node] (a) at (1, {2*cos(30)} )  {$a$ } ;
		\node[] (text) at (1,{2*cos(30)+1}) {\large{$G$}};
		
		\draw[very thick,-,shorten >= 3pt,shorten <= 3pt]
		(b) edge[->] (c)
		(c) edge[->] (a)
		(a) edge[->] (b)
		;
		
		\node[main node] (bc) at (5,0){$bc$} ;
		\node[main node] (ca) at (6,{2*cos(30)} ) { $ca$} ;
		\node[main node] (ab) at (4,{2*cos(30)} ) {$ab$ } ;
		\node[] (text) at (5,{2*cos(30)+1}) {\large{$\mathcal L(G)$}};
		
		\draw[very thick,-,shorten >= 3pt,shorten <= 3pt]
		(ca) edge[->] (ab)
		(bc) edge[->] (ca)
		(ab) edge[->] (bc)
		;
		\end{tikzpicture}
		\\
		\vspace{.25 cm}
		\begin{tikzpicture}[main node/.style={circle, draw,thick, inner sep=1pt,minimum size=18pt}, scale=.9,node distance=1cm]
		\tikzstyle{every loop}=[looseness=14]
		
		\node[main node] (b) at (0,0) {$b$ };
		\node[main node] (c) at (2,0)  {$c$ };
		\node[main node] (a) at (1, {2*cos(30)} )  {$a$ } ;
		
		\draw[very thick,-,shorten >= 3pt,shorten <= 3pt]
		(b) edge[<-] (c)
		(c) edge[->] (a)
		(a) edge[->] (b)
		;
		
		\node[main node] (bc) at (5,0){$bc$} ;
		\node[main node] (ca) at (6,{2*cos(30)} ) { $ca$} ;
		\node[main node] (ab) at (4,{2*cos(30)} ) {$ab$ } ;

		\draw[very thick,-,shorten >= 3pt,shorten <= 3pt]
		(bc) edge[->,dotted] (ca)
		(ca) edge[->,dotted] (bc)
		(ca) edge[->] (ab)
		(ab) edge[->,dashed] (bc)
		(bc) edge[->,dashed] (ab)
		;
		
		\end{tikzpicture}\\
		\vspace{.25 cm}
		\begin{tikzpicture}[main node/.style={circle, draw,thick, inner sep=1pt,minimum size=18pt}, scale=.9,node distance=1cm]
		\tikzstyle{every loop}=[looseness=14]

		\node[main node] (a) at (0,0) {$a$} ;
		\node[main node] (b) [below=of a] {$d$} ;
		\node[main node] (c) [above right=of a] {$c$} ;
		\node[main node] (d) [above left=of a] {$b$} ;

		\draw[very thick,-,shorten >= 3pt,shorten <= 3pt]
		(b) edge[->] (a)
		(c) edge[->] (a)
		(d) edge[->] (a)
		;
		
		\node[main node] (da) at (4,{-cos(30)}){$da$} ;
		\node[main node] (ca) at (5,{2*cos(30)- cos(30)} ) { $ca$} ;
		\node[main node] (ba) at (3,{2*cos(30) - cos(30)} ) {$ba$ } ;

		\draw[very thick,-,shorten >= 3pt,shorten <= 3pt]
		(ba) edge[->,dashed] (ca)
		(ca) edge[->,dashed] (ba)
		(ca) edge[->,dashed] (da)
		(da) edge[->,dashed] (ca)
		(da) edge[->,dashed] (ba)
		(ba) edge[->,dashed] (da)
		;
		
		\end{tikzpicture}\\
		\vspace{.25 cm}
		\begin{tikzpicture}[main node/.style={circle, draw,thick, inner sep=1pt,minimum size=18pt}, scale=.9,node distance=1cm]
		\tikzstyle{every loop}=[looseness=14]

		\node[main node] (a) at (0,0) {$a$} ;
		\node[main node] (d) [below=of a] {$d$} ;
		\node[main node] (c) [above right=of a] {$c$} ;
		\node[main node] (b) [above left=of a] {$b$} ;

		\draw[very thick,-,shorten >= 3pt,shorten <= 3pt]
		(b) edge[->] (a)
		(c) edge[->] (a)
		(a) edge[->] (d)
		;
		
		\node[main node] (da) at (4,{-cos(30)}){$da$} ;
		\node[main node] (ca) at (5,{2*cos(30)- cos(30)} ) { $ca$} ;
		\node[main node] (ba) at (3,{2*cos(30) - cos(30)} ) {$ba$ } ;

		\draw[very thick,-,shorten >= 3pt,shorten <= 3pt]
		(ba) edge[->,dashed] (ca)
		(ca) edge[->,dashed] (ba)
		(ca) edge[->] (da)
		(ba) edge[->] (da)
		;
		
		\end{tikzpicture}\\
		\vspace{.25 cm}
		\begin{tikzpicture}[main node/.style={circle, draw,thick, inner sep=1pt,minimum size=18pt}, scale=.9,node distance=1cm]
		\tikzstyle{every loop}=[looseness=14]

		\node[main node] (a) at (0,0) {$a$} ;
		\node[main node] (d) [below=of a] {$d$} ;
		\node[main node] (c) [above right=of a] {$c$} ;
		\node[main node] (b) [above left=of a] {$b$} ;

		\draw[very thick,-,shorten >= 3pt,shorten <= 3pt]
		(b) edge[->] (a)
		(a) edge[->] (c)
		(a) edge[->] (d)
		;
		
		\node[main node] (da) at (4,{-cos(30)}){$da$} ;
		\node[main node] (ca) at (5,{2*cos(30)- cos(30)} ) { $ca$} ;
		\node[main node] (ba) at (3,{2*cos(30) - cos(30)} ) {$ba$ } ;

		\draw[very thick,-,shorten >= 3pt,shorten <= 3pt]
		(ba) edge[->] (ca)
		(ca) edge[->,dotted] (da)
		(da) edge[->,dotted] (ca)
		(ba) edge[->] (da)
		;
		
		\end{tikzpicture}\\
		\vspace{.25 cm}
		\begin{tikzpicture}[main node/.style={circle, draw,thick, inner sep=1pt,minimum size=18pt}, scale=.9,node distance=1cm]
		\tikzstyle{every loop}=[looseness=14]

		\node[main node] (a) at (0,0) {$a$} ;
		\node[main node] (d) [below=of a] {$d$} ;
		\node[main node] (c) [above right=of a] {$c$} ;
		\node[main node] (b) [above left=of a] {$b$} ;

		\draw[very thick,-,shorten >= 3pt,shorten <= 3pt]
		(a) edge[->] (b)
		(a) edge[->] (c)
		(a) edge[->] (d)
		;
		
		\node[main node] (da) at (4,{-cos(30)}){$da$} ;
		\node[main node] (ca) at (5,{2*cos(30)- cos(30)} ) { $ca$} ;
		\node[main node] (ba) at (3,{2*cos(30) - cos(30)} ) {$ba$ } ;

		\draw[very thick,-,shorten >= 3pt,shorten <= 3pt]
		(ba) edge[->,dotted] (ca)
		(ca) edge[->,dotted] (da)
		(da) edge[->,dotted] (ba)
		(ca) edge[->,dotted] (ba)
		(da) edge[->,dotted] (ca)
		(ba) edge[->,dotted] (da)
		;
		
		\end{tikzpicture}\\
		\vspace{.25 cm}
	\end{center}
	\caption{The extended line digraphs of graphs structurally isomorphic to $\Delta$ or $Y$ (up to node relabeling). For an extended line digraph $\mathcal L(G) = (\mathcal D, \mathcal D_L, \bar \ell_v,\bar \ell_e)$ in the right column, solid lines indicate an edge label of $\bar \ell_e(\{e,e'\}) = ht$ (head-to-tail relationship), dashed lines indicate $\bar \ell_e(\{e,e'\}) = hh$ (head-to-head) and dotted lines indicate $\bar \ell_e(\{e,e'\}) = tt$ (tail-to-tail).
	  Note that the letters appearing in nodes are not necessarily labels. They are  used here as node and edge identities to relate $G$ and $\mathcal L(G)$.}
\end{figure}

 
 The next result  establishes that isomorphism between extended line digraphs implies structural isomorphism between digraphs.
 
	\begin{lemma}\label{lem:line_implies_structure}
	 Given two node-labeled digraphs $G,G'$ satisfying \textbf{(W)}, \textbf{(S)}, and \textbf{(O)}, if $\mathcal L(G)\cong \mathcal L(G')$, then $S(G) \cong S(G')$. 
	\end{lemma}
	\begin{proof}
$\mathcal L(G)\cong \mathcal L(G')$ implies $S(\mathcal L(G)) \cong S(\mathcal L(G'))$.  By Lemma~\ref{lem:flip}, it follows that  $L(S(G)) \cong L(S(G'))$. 
 Next,  by the contrapositive of Lemma~\ref{lem:delta_y}, $\mathcal L(G)\cong \mathcal L(G')$ implies either $S(G) \cong S(G') \cong \Delta$, $S(G) \cong S(G') \cong Y$, or at least one of  $S(G), S(G')$ is not isomorphic to either $Y$ or  $\Delta$. In either of the first two cases the proof is complete so assume at least one of  $S(G), S(G')$ is not isomorphic to either $Y$ or  $\Delta$. Since $G$ and $G'$ are weakly connected by assumption, then $S(G)$ and $S(G')$ are connected. We may then directly apply the Whitney Graph Isomorphism Theorem \cite{whitney} to show that  $S(G) \cong S(G')$.
\end{proof}

		\begin{corollary}\label{cor:Y}	Let $G = (V,\mathcal D,\ell_v,\emptyset)$ be a node-labeled digraph satisfying \textbf{(W)}, \textbf{(S)}, and \textbf{(O)}, that further has a subgraph $H = (U,W,\ell_v|_U,\emptyset)$ such that $S(H)$ is isomorphic to the $Y$ graph. Let $G' = (V',\mathcal D',\ell_v',\emptyset)$ be a node-labeled digraph satisfying \textbf{(W)}, \textbf{(S)}, and \textbf{(O)} such that $\phi: \mathcal D \to \mathcal D'$ is an isomorphism between the extended line digraphs $\mathcal L(G)$ and $\mathcal L(G')$.   
		Then the edge-induced subgraph $H'$ of $G'$, induced by the set of edges $\phi(W)$, has a structure isomorphic to the $Y$ graph, $S(H') \cong Y$.
		\end{corollary}
		
	\begin{proof}
	We apply Lemma~\ref{lem:line_implies_structure} to  $\phi|_W$, which is  an isomorphism between $\mathcal L(H)$ and $\mathcal L(H')$.  Then  $S(H) \cong S(H')$ follows.
	\end{proof}

\begin{definition}
	The \emph{degree of a node $v$} in a digraph $G$, denoted $\text{deg } v$,  is the number of incoming and outgoing edges incident to the node $v$.
\end{definition}

\paragraph{Proof of Theorem~\ref{thm:iso_iff}}

Let $G=(V,\mathcal D, \ell_v,\emptyset)$ and $G'=(V',\mathcal D', \ell'_v,\emptyset)$ be node-labeled digraphs satisfying \textbf{(W)}, \textbf{(S)}, and \textbf{(O)}. Denote $\mathcal L (G) = (\mathcal D,\mathcal D_L,\bar \ell_v,\bar \ell_e)$  and  $\mathcal L (G') = (\mathcal D',\mathcal D_L',\bar \ell'_v,\bar \ell'_e)$ as the extended line digraphs of $G$ and $G'$. 
We seek to prove that $\mathcal L(G) \cong \mathcal L(G')$ if and only if $G \cong G'$.
	
The reverse direction, $G \cong G'$ implies $\mathcal L(G) \cong \mathcal L(G')$, is nearly immediate. Let $\gamma : V \to V'$ be the isomorphism between $G$ and $G'$. For an edge $(v_1,v_2) \in \mathcal D$, construct the function  $\phi : \mathcal D \to \mathcal D'$ by
\begin{equation}\label{eq:iso2iso}
	\phi((v_1,v_2)) \mapsto (\gamma(v_1),\gamma(v_2)).
\end{equation}
Since $\gamma$ is a bijection, $\phi$ is a bijection that clearly respects the orientation of the edges from $G$ to $G'$. Moreover, since $\gamma$ is label-respecting, $\phi$ is label-respecting as well. Thus $\phi$ is an isomorphism between the extended line digraphs.
	
	Now consider the forward direction, $\mathcal L(G) \cong \mathcal L(G')$ implies $G \cong G'$. Let $\phi : \mathcal D \to \mathcal D'$ be the isomorphism between  $\mathcal L(G) $ and $ \mathcal L(G')$. We will show that from $\phi$ one can construct a function $\gamma$ satisfying~\eqref{eq:iso2iso} that is an isomorphism between $G$ and $G'$. The theorem is easy to verify for $|V|,|V'| \leq 2$ so assume $|V|,|V'| >2$. We remark that our construction of $\gamma$ follows the outline of a proof  due to \cite{jung}, given in \cite{harary}. 
	
	Consider a vertex $v \in V$ and let $P(v) \subseteq \mathcal D$ be the set of edges incident on $v$, i.e. $$P(v) =\{(u,w) \in \mathcal D\mid u=v\text{ or }w=v\},$$ so that $|P(v)| = \text{deg }v$. Since $G,G'$ are simple and oriented, for $u \neq v$ we have that $P(u) \cap P(v)$ contains at most one edge. Also by necessity, $P(u) \cap P(v) \cap P(w) = \emptyset$ whenever $u,v,w$ are distinct.
	
First suppose $\text{deg }v > 1$. Let $e_1$ and $e_2$ be two edges connected to $v$ in $G$.  Then $ e_1, e_2 \in \mathcal D$ is a pair of directed edges that have either a head-to-tail (label $ht$), head-to-head, ($hh$) or tail-to-tail ($tt$) relationship. Since $\phi$ is an isomorphism,
	$\phi(e_1), \phi(e_2) \in \mathcal D'$ share some node $v' \in G'$, which is to say $\phi(e_1),\phi(e_2) \in P(v')$.  Since $G'$ is simple and oriented, $\phi(e_1)$ and $\phi(e_2)$ can share a maximum of one node, so $v'$ is uniquely determined by the isomorphism $\phi$.
	
	Now assume there is another edge $e_3\neq e_1, e_2$ connected to $v$. Then $\phi(e_1)$ and $\phi(e_3)$ form a pair of edges in $\mathcal D'$
	that share a node $v''$.  Similarly, $\phi(e_2), \phi(e_3)$ form a pair of edges that share a node $v'''$. Let $H$ be the edge-induced subgraph of $G$ induced by  $\{e_1,e_2,e_3\}$. Then the structure of $H$, $S(H)$, is isomorphic to the $Y$ graph shown in Figure~\ref{fig:dY}, where the degree three node in the middle is $v$.
	By Corollary~\ref{cor:Y}, the $\{\phi(e_1),\phi(e_2), \phi(e_3) \}$ edge-induced subgraph of $G'$, say $H'$, must also satisfy $S(H') \cong Y$. This implies that 
	\[ v'=v''=v''' \]
	is the degree three node  in $H'$. Therefore, $\phi(P(v)) \subseteq P(v')$.  For the same reason, for any edge $e' \neq \phi(e_1) , \phi(e_2)$ connected to $v'$ the $\{e', \phi(e_1),\phi(e_2)\}$ edge-induced subgraph of $G'$ has structure isomorphic to $Y$ and so the $\{\phi^{-1}(e'), e_1,e_2\}$ edge-induced subgraph of $G$ also has structure isomorphic to $Y$, implying  $\phi^{-1} (e') \in P(v)$. Thus, $\phi(P(v)) = P(v')$ and $v'$ is uniquely determined by $\phi$. 
	We can then define the injection
	$$\gamma|_W: W \to V'$$
	$$v \mapsto v'$$
	where $W \subseteq V$ is the subset of nodes of $V$ with degree greater than 1 and $v'$ is the unique node in $V'$ such that $\phi(P(v)) = P(v')$.

%
%

	Next suppose $\text{deg }v = 1$. Let $u$ be the single neighbor of $v$ and let $e_1$ be the directed edge connecting $u$ and $v$.  
	Since the digraphs are weakly connected and we assume that the number of vertices of $G$ is  greater than 2, $\text{deg }u >1$. Then $\gamma|_W$ is well defined on $u$ and we let $u'= \gamma|_W(u)$. Then $\phi(e_1) \in P(u')$ and we let $v'$ be the other node of the edge $\phi(e_1)$.  
	
	We now show that $\text{deg }v' = 1$. 
	By contradiction, assume $\text{deg }v' >1$. Then $w:= \gamma|_W^{-1}(v')$ is a vertex in $V$ with $\text{deg }w >1$, which implies $w\not = v$. 
	Now $\phi(e_1) \in P(u') \cap P(v')$ implies that $e_1 \in P(u) \cap P(w)$. We already know that $e_1 \in P(u) \cap P(v)$ as well. Therefore $e_1 \in P(u) \cap P(v) \cap P(w) = \emptyset$, an impossible condition. We conclude $\text{deg }v' = 1$, with $\phi(P(v)) = P(v')$, where $v'$ is uniquely determined.
We therefore extend $\gamma|_W$ to the injection
	 $$\gamma: V \to V'$$
	 $$v \mapsto v'.$$
	 
	 The map $\gamma$ is in fact a bijection. To see this, assume by contradiction that $\gamma$ is not surjective. Then there exists a $v' \in V'$ such that $\gamma^{-1}(v')$ does not exist. Since $\phi^{-1}$ is an isomorphism, this means that $v'$ participates in no edges in $\mathcal D'$; i.e. $v'$ is an isolated node. But this contradicts the fact that $G'$ is weakly connected, so $\gamma$ must be surjective.

	The condition $\phi(P(v)) = P(\gamma(v))$ shows that the connectivity of the graph $G$ is conserved under $\gamma$ in $G'$. Now we need to consider the orientation of these edges. Consider a directed edge $e_1:= (u, v) \in \mathcal D$ connecting two nodes $u,v \in G$. Since $G$ is weakly connected with $|V| >2$, either $u$ or $v$ has degree greater than one.

	Assume first that $\text{deg }v >1$. Then there is an edge $e_2 \neq e_1$ incident on $v$, connecting $v$ and $w$, so that $e_2 = (v,w) \text{ or } (w,v)$ and $\phi(e_2) =  (\gamma(v),\gamma(w)) \text{ or } (\gamma(w),\gamma(v))$.  Consider any edge in $\mathcal D_L$  connecting $e_1$ and  $e_2$. 
	If $(e_1,e_2) \in \mathcal D_L$, then either $e_2 = (w,v)$ and $\bar \ell_e((e_1,e_2)) = hh$ or $e_2 = (v,w)$ and $\bar \ell_e((e_1,e_2)) = ht$.
	If the label is $hh$, then the edge $(e_2,e_1) \in \mathcal D_L$ with $\bar \ell_e((e_2,e_1)) = hh$ necessarily and vice versa. 
	The possible edge $(e_2,e_1) = ((v,w),(u,v))$ is never constructed in $\mathcal L(G)$, since it is a tail-to-head relationship.
	So by necessity, $q = (e_1,e_2) \in \mathcal D_L$ and the label $\bar \ell_v(q)$ is sufficient to capture the relationship of $e_1,e_2$ as either head-to-head or head-to-tail.
	
	Now consider $\phi(e_1), \phi(e_2)$. From the observation above and the fact that $\phi$ is an isomorphism, we have that $q' = (\phi(e_1),\phi(e_2)) \in \mathcal D_L'$ and $\bar \ell_e'(q') = hh$ or $ht$ are the only two possibilities. If $\bar \ell_e'(q') = hh$, then necessarily $((\gamma(u),\gamma(v)), (\gamma(w),\gamma(v))) \in \mathcal D_L'$, since $\gamma(v)$ is the only common node and must occur in a head-to-head relationship. For $\bar \ell_e'(q') = ht$, we must have that $((\gamma(u),\gamma(v)), (\gamma(v),\gamma(w))) \in \mathcal D_L'$. 
	Since $\phi((u,v)) = (\gamma(u),\gamma(v))$ in both cases, $\phi((v,w)) = (\gamma(v),\gamma(w))$ when the label is $ht$, and $\phi((w,v)) = (\gamma(w),\gamma(v))$ when the label is $hh$, we conclude that $\gamma$ conserves the orientation of the edges when $\text{deg }v >1$.

A similar argument holds when $\text{deg }u > 1$. In this case, the possible labels are $tt$ and $ht$. With these arguments, we have shown that
\begin{equation}\label{eq:phigamma} 
	\phi(e) = \phi((u,v)) = (\gamma(u),\gamma(v)) \quad \forall (u,v) \in \mathcal D, \end{equation}
so that the orientation of the edges in $G$ is preserved under the bijection $\gamma$.

The last step is to show that node labels are respected between $G$ and $G'$. Let $(u,v) \in \mathcal D$ and recall that the label of $(u,v)$ in the extended line digraph is $(\ell_v(u),\ell_v(v))$. Now,
\[ (\ell_v(u),\ell_v(v)) = \bar \ell_v ((u,v)) = \bar \ell_v' (\phi((u,v))) =  \bar \ell_v' ((\gamma(u),\gamma(v))) = (\ell _v'(\gamma(u)),\ell_v'(\gamma(v))). \]
The first equality follows from the definition of the node labels in the extended line digraph. The second equality follows from the preservation of labels under the isomorphism $\phi$. The third equality follows from~\eqref{eq:phigamma} and the last is again the definition of node labeling in extended line digraphs. The chain of equalities allows us to state that $\ell_v(u) = \ell_v'(\gamma(u))$ for all $u \in V$, so that $\gamma$ is a label-respecting bijection that preserves edge orientation and is therefore an isomorphism between $G$ and $G'$. This completes the proof that $\mathcal L(G) \cong \mathcal L(G')$ if and only if $G \cong G'$.

\section{Figures for Lemma~\ref{lem:alwaysIncrease}}\label{app:figs}
The following figure shows  all possible arrangements of the subgraphs induced by $U_{(x)} = \{u,v,x\}$, and the corresponding images $\phi(U_{(x)})$ and $\psi(U_{(x)})$ under the assumptions that $G$ and $G'$ are oriented, transitive closures. We use the notation $x' = \phi(x) = \psi(x) $, $v' = \psi(v) $, $u' = \psi(u)$, $v' = \phi(u) $, and $u' = \phi(v)$. Under each graph we list the score of both feasible solutions $(U_{(x)},\phi)$ and $(U_{(x)},\psi)$, and the cardinality of the set $\mathcal X$ in both cases. For ease of notation, we write  $\phi|_{U_{(x)}}$ as $\phi$, and similarly $\psi|_{U_{(x)}}$ as $\psi$. In each of the following cases, the Equation~\eqref{eq:alwaysIncrease} is verified. 
\newpage

	\begin{multicols}{3}
		\setlength{\columnseprule}{0.4pt}
		
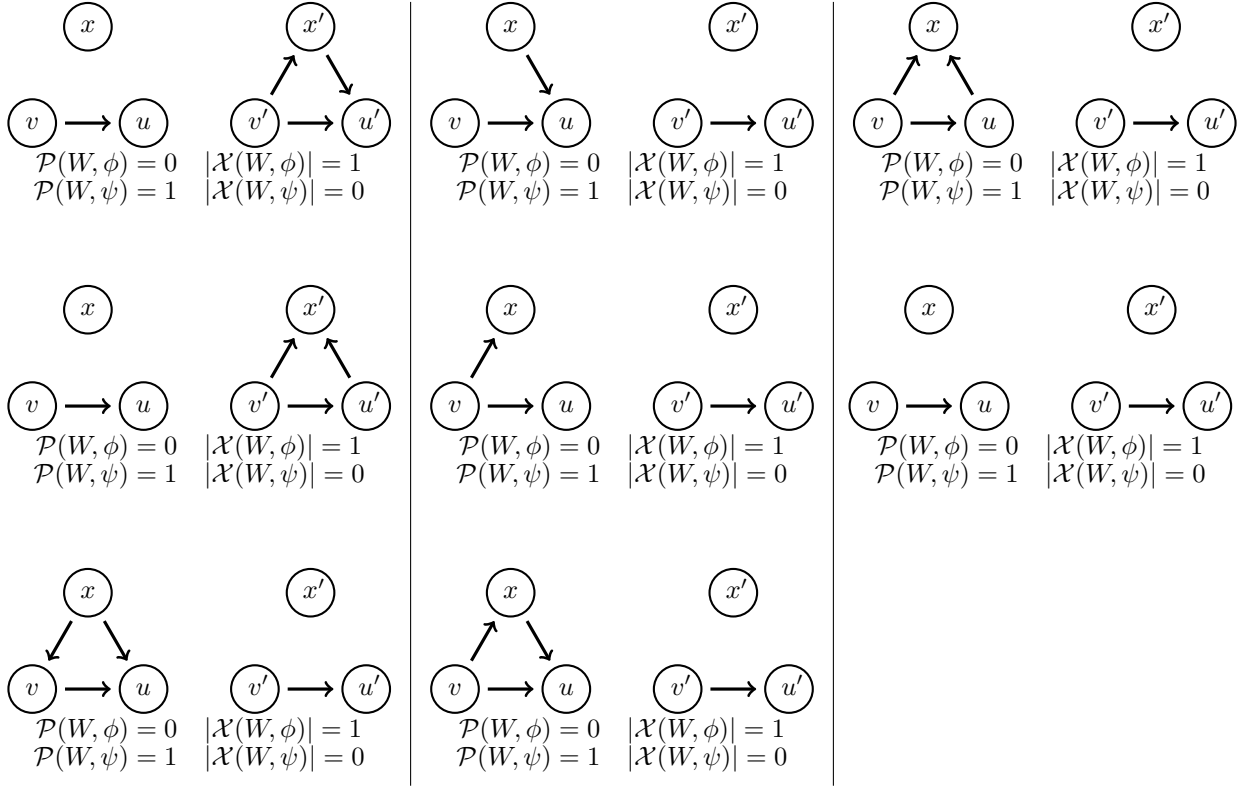
		\begin{center}
		
			\begin{tikzpicture}[main node/.style={circle, draw,thick, inner sep=1pt,minimum size=18pt}, scale=.74,node distance=1cm]
			\tikzstyle{every loop}=[looseness=14]
			
			\node[main node] (v) at (0,0) {$v$} ;
			\node[main node] (u) at (2,0) {$u$} ;
			\node[main node] (x) at (1, {2*cos(30)} ) {$x$} ;
			
			\draw[very thick,-,shorten >= 3pt,shorten <= 3pt]
			(v) edge[->] (u)
			(x) edge[->] (v)
			(x) edge[->] (u)
			;
			
			\node[main node] (v') at (4,0) {$v\ghostprime$} ;
			\node[main node] (u') at (6,0) {$u\ghostprime$} ;
			\node[main node] (x') at (5, {2*cos(30)} ) {$x\ghostprime$} ;
			\node[] (text') at (3,-.75) {$\mathcal P(U_{(x)},\phi)  = 2 \quad |\mathcal X(U_{(x)},\phi)| = 1$};
			\node[] (text') at (3,-1.25) {$\mathcal P(U_{(x)},\psi)  = 3\quad |\mathcal X(U_{(x)},\psi)| = 0$};
			
			\draw[very thick,-,shorten >= 3pt,shorten <= 3pt]
			(v') edge[->] (u')
			(x') edge[->] (v')
			(x') edge[->] (u')
			;
			\end{tikzpicture}
			
			\vspace{.9cm}
			\begin{tikzpicture}[main node/.style={circle, draw,thick, inner sep=1pt,minimum size=18pt}, scale=.74,node distance=1cm]
			\tikzstyle{every loop}=[looseness=14]
			
			\node[main node] (v) at (0,0) {$v$} ;
			\node[main node] (u) at (2,0) {$u$} ;
			\node[main node] (x) at (1, {2*cos(30)} ) {$x$} ;
			
			\draw[very thick,-,shorten >= 3pt,shorten <= 3pt]
			(v) edge[->] (u)
			(x) edge[->] (v)
			(x) edge[->] (u)
			;
			
			\node[main node] (v') at (4,0) {$v\ghostprime$} ;
			\node[main node] (u') at (6,0) {$u\ghostprime$} ;
			\node[main node] (x') at (5, {2*cos(30)} ) {$x\ghostprime$} ;
			\node[] (text') at (3,-.75) {$\mathcal P(U_{(x)},\phi)  = 1  \quad |\mathcal X(U_{(x)},\phi)| = 1$};
			\node[] (text') at (3,-1.25) {$\mathcal P(U_{(x)},\psi)  = 2 \quad |\mathcal X(U_{(x)},\psi)| = 0$};
			
			\draw[very thick,-,shorten >= 3pt,shorten <= 3pt]
			(v') edge[->] (u')
			(x') edge[->] (u')
			;
			\end{tikzpicture}
			
			\vspace{.9cm}
			\begin{tikzpicture}[main node/.style={circle, draw,thick, inner sep=1pt,minimum size=18pt}, scale=.74,node distance=1cm]
			\tikzstyle{every loop}=[looseness=14]
			
			\node[main node] (v) at (0,0) {$v$} ;
			\node[main node] (u) at (2,0) {$u$} ;
			\node[main node] (x) at (1, {2*cos(30)} ) {$x$} ;
			
			\draw[very thick,-,shorten >= 3pt,shorten <= 3pt]
			(v) edge[->] (u)
			(x) edge[->] (v)
			(x) edge[->] (u)
			;
			
			\node[main node] (v') at (4,0) {$v\ghostprime$} ;
			\node[main node] (u') at (6,0) {$u\ghostprime$} ;
			\node[main node] (x') at (5, {2*cos(30)} ) {$x\ghostprime$} ;
			\node[] (text') at (3,-.75) {$\mathcal P(U_{(x)},\phi)  = 0  \quad |\mathcal X(U_{(x)},\phi)| = 2$};
			\node[] (text') at (3,-1.25) {$\mathcal P(U_{(x)},\psi)  = 1 \quad |\mathcal X(U_{(x)},\psi)| = 1$};
			
			\draw[very thick,-,shorten >= 3pt,shorten <= 3pt]
			(v') edge[->] (u')
			(x') edge[<-] (v')
			;
			\end{tikzpicture}
			
			\vspace{.9cm}
			\begin{tikzpicture}[main node/.style={circle, draw,thick, inner sep=1pt,minimum size=18pt}, scale=.74,node distance=1cm]
			\tikzstyle{every loop}=[looseness=14]
			
			\node[main node] (v) at (0,0) {$v$} ;
			\node[main node] (u) at (2,0) {$u$} ;
			\node[main node] (x) at (1, {2*cos(30)} ) {$x$} ;
			
			\draw[very thick,-,shorten >= 3pt,shorten <= 3pt]
			(v) edge[->] (u)
			(x) edge[->] (v)
			(x) edge[->] (u)
			;
			
			\node[main node] (v') at (4,0) {$v\ghostprime$} ;
			\node[main node] (u') at (6,0) {$u\ghostprime$} ;
			\node[main node] (x') at (5, {2*cos(30)} ) {$x\ghostprime$} ;
			\node[] (text') at (3,-.75) {$\mathcal P(U_{(x)},\phi)  = 1  \quad |\mathcal X(U_{(x)},\phi)| = 2$};
			\node[] (text') at (3,-1.25) {$\mathcal P(U_{(x)},\psi)  = 2 \quad |\mathcal X(U_{(x)},\psi)| = 1$};
			
			\draw[very thick,-,shorten >= 3pt,shorten <= 3pt]
			(v') edge[->] (u')
			(x') edge[<-] (v')
			(x') edge[->] (u')
			;
			\end{tikzpicture}
			
			\vspace{.9cm}
			\begin{tikzpicture}[main node/.style={circle, draw,thick, inner sep=1pt,minimum size=18pt}, scale=.74,node distance=1cm]
			\tikzstyle{every loop}=[looseness=14]
			
			\node[main node] (v) at (0,0) {$v$} ;
			\node[main node] (u) at (2,0) {$u$} ;
			\node[main node] (x) at (1, {2*cos(30)} ) {$x$} ;
			
			\draw[very thick,-,shorten >= 3pt,shorten <= 3pt]
			(v) edge[->] (u)
			(x) edge[->] (v)
			(x) edge[->] (u)
			;
			
			\node[main node] (v') at (4,0) {$v\ghostprime$} ;
			\node[main node] (u') at (6,0) {$u\ghostprime$} ;
			\node[main node] (x') at (5, {2*cos(30)} ) {$x\ghostprime$} ;
			\node[] (text') at (3,-.75) {$\mathcal P(U_{(x)},\phi)  = 0  \quad |\mathcal X(U_{(x)},\phi)| = 3$};
			\node[] (text') at (3,-1.25) {$\mathcal P(U_{(x)},\psi)  = 1 \quad |\mathcal X(U_{(x)},\psi)| = 2$};
			
			\draw[very thick,-,shorten >= 3pt,shorten <= 3pt]
			(v') edge[->] (u')
			(x') edge[<-] (v')
			(x') edge[<-] (u')
			;
			\end{tikzpicture}
			
			\vspace{.9cm}
			\begin{tikzpicture}[main node/.style={circle, draw,thick, inner sep=1pt,minimum size=18pt}, scale=.74,node distance=1cm]
			\tikzstyle{every loop}=[looseness=14]
			
			\node[main node] (v) at (0,0) {$v$} ;
			\node[main node] (u) at (2,0) {$u$} ;
			\node[main node] (x) at (1, {2*cos(30)} ) {$x$} ;
			
			\draw[very thick,-,shorten >= 3pt,shorten <= 3pt]
			(v) edge[->] (u)
			(x) edge[<-] (v)
			(x) edge[<-] (u)
			;
			
			\node[main node] (v') at (4,0) {$v\ghostprime$} ;
			\node[main node] (u') at (6,0) {$u\ghostprime$} ;
			\node[main node] (x') at (5, {2*cos(30)} ) {$x\ghostprime$} ;
			\node[] (text') at (3,-.75) {$\mathcal P(U_{(x)},\phi)  = 0  \quad |\mathcal X(U_{(x)},\phi)| = 2$};
			\node[] (text') at (3,-1.25) {$\mathcal P(U_{(x)},\psi)  = 1 \quad |\mathcal X(U_{(x)},\psi)| = 1$};
			
			\draw[very thick,-,shorten >= 3pt,shorten <= 3pt]
			(v') edge[->] (u')
			(x') edge[->] (u')
			;
			\end{tikzpicture}
			
			\vspace{.9cm}
			\begin{tikzpicture}[main node/.style={circle, draw,thick, inner sep=1pt,minimum size=18pt}, scale=.74,node distance=1cm]
			\tikzstyle{every loop}=[looseness=14]
			
			\node[main node] (v) at (0,0) {$v$} ;
			\node[main node] (u) at (2,0) {$u$} ;
			\node[main node] (x) at (1, {2*cos(30)} ) {$x$} ;
			
			\draw[very thick,-,shorten >= 3pt,shorten <= 3pt]
			(v) edge[->] (u)
			(x) edge[<-] (v)
			(x) edge[<-] (u)
			;
			
			\node[main node] (v') at (4,0) {$v\ghostprime$} ;
			\node[main node] (u') at (6,0) {$u\ghostprime$} ;
			\node[main node] (x') at (5, {2*cos(30)} ) {$x\ghostprime$} ;
			\node[] (text') at (3,-.75) {$\mathcal P(U_{(x)},\phi)  = 1  \quad |\mathcal X(U_{(x)},\phi)| = 1$};
			\node[] (text') at (3,-1.25) {$\mathcal P(U_{(x)},\psi)  = 2 \quad |\mathcal X(U_{(x)},\psi)| = 0$};
			
			\draw[very thick,-,shorten >= 3pt,shorten <= 3pt]
			(v') edge[->] (u')
			(x') edge[<-] (v')
			;
			\end{tikzpicture}
			
			\vspace{.9cm}
			\begin{tikzpicture}[main node/.style={circle, draw,thick, inner sep=1pt,minimum size=18pt}, scale=.74,node distance=1cm]
			\tikzstyle{every loop}=[looseness=14]
			
			\node[main node] (v) at (0,0) {$v$} ;
			\node[main node] (u) at (2,0) {$u$} ;
			\node[main node] (x) at (1, {2*cos(30)} ) {$x$} ;
			
			\draw[very thick,-,shorten >= 3pt,shorten <= 3pt]
			(v) edge[->] (u)
			(x) edge[<-] (v)
			(x) edge[<-] (u)
			;
			
			\node[main node] (v') at (4,0) {$v\ghostprime$} ;
			\node[main node] (u') at (6,0) {$u\ghostprime$} ;
			\node[main node] (x') at (5, {2*cos(30)} ) {$x\ghostprime$} ;
			\node[] (text') at (3,-.75) {$\mathcal P(U_{(x)},\phi)  = 1 \quad |\mathcal X(U_{(x)},\phi)| = 2$};
			\node[] (text') at (3,-1.25) {$\mathcal P(U_{(x)},\psi)  = 2 \quad |\mathcal X(U_{(x)},\psi)| = 1$};
			
			\draw[very thick,-,shorten >= 3pt,shorten <= 3pt]
			(v') edge[->] (u')
			(x') edge[<-] (v')
			(x') edge[->] (u')
			;
			\end{tikzpicture}
			
			\vspace{.9cm}
			\begin{tikzpicture}[main node/.style={circle, draw,thick, inner sep=1pt,minimum size=18pt}, scale=.74,node distance=1cm]
			\tikzstyle{every loop}=[looseness=14]
			
			\node[main node] (v) at (0,0) {$v$} ;
			\node[main node] (u) at (2,0) {$u$} ;
			\node[main node] (x) at (1, {2*cos(30)} ) {$x$} ;
			
			\draw[very thick,-,shorten >= 3pt,shorten <= 3pt]
			(v) edge[->] (u)
			(x) edge[<-] (v)
			(x) edge[<-] (u)
			;
			
			\node[main node] (v') at (4,0) {$v\ghostprime$} ;
			\node[main node] (u') at (6,0) {$u\ghostprime$} ;
			\node[main node] (x') at (5, {2*cos(30)} ) {$x\ghostprime$} ;
			\node[] (text') at (3,-.75) {$\mathcal P(U_{(x)},\phi)  = 2  \quad |\mathcal X(U_{(x)},\phi)| = 1$};
			\node[] (text') at (3,-1.25) {$\mathcal P(U_{(x)},\psi)  = 3 \quad |\mathcal X(U_{(x)},\psi)| = 0$};
			
			\draw[very thick,-,shorten >= 3pt,shorten <= 3pt]
			(v') edge[->] (u')
			(x') edge[<-] (v')
			(x') edge[<-] (u')
			;
			\end{tikzpicture}
			
			\vspace{.9cm}
			\begin{tikzpicture}[main node/.style={circle, draw,thick, inner sep=1pt,minimum size=18pt}, scale=.74,node distance=1cm]
			\tikzstyle{every loop}=[looseness=14]
			
			\node[main node] (v) at (0,0) {$v$} ;
			\node[main node] (u) at (2,0) {$u$} ;
			\node[main node] (x) at (1, {2*cos(30)} ) {$x$} ;
			
			\draw[very thick,-,shorten >= 3pt,shorten <= 3pt]
			(v) edge[->] (u)
			(x) edge[->] (u)
			;
			
			\node[main node] (v') at (4,0) {$v\ghostprime$} ;
			\node[main node] (u') at (6,0) {$u\ghostprime$} ;
			\node[main node] (x') at (5, {2*cos(30)} ) {$x\ghostprime$} ;
			\node[] (text') at (3,-.75) {$\mathcal P(U_{(x)},\phi)  = 0  \quad |\mathcal X(U_{(x)},\phi)| = 1$};
			\node[] (text') at (3,-1.25) {$\mathcal P(U_{(x)},\psi)  = 2 \quad |\mathcal X(U_{(x)},\psi)| = 0$};
			
			\draw[very thick,-,shorten >= 3pt,shorten <= 3pt]
			(v') edge[->] (u')
			(x') edge[->] (u')
			;
			\end{tikzpicture}
			
			\vspace{.9cm}
			\begin{tikzpicture}[main node/.style={circle, draw,thick, inner sep=1pt,minimum size=18pt}, scale=.74,node distance=1cm]
			\tikzstyle{every loop}=[looseness=14]
			
			\node[main node] (v) at (0,0) {$v$} ;
			\node[main node] (u) at (2,0) {$u$} ;
			\node[main node] (x) at (1, {2*cos(30)} ) {$x$} ;
			
			\draw[very thick,-,shorten >= 3pt,shorten <= 3pt]
			(v) edge[->] (u)
			(x) edge[->] (u)
			;
			
			\node[main node] (v') at (4,0) {$v\ghostprime$} ;
			\node[main node] (u') at (6,0) {$u\ghostprime$} ;
			\node[main node] (x') at (5, {2*cos(30)} ) {$x\ghostprime$} ;
			\node[] (text') at (3,-.75) {$\mathcal P(U_{(x)},\phi)  = 0  \quad |\mathcal X(U_{(x)},\phi)| = 2$};
			\node[] (text') at (3,-1.25) {$\mathcal P(U_{(x)},\psi)  = 1 \quad |\mathcal X(U_{(x)},\psi)| = 0$};
			
			\draw[very thick,-,shorten >= 3pt,shorten <= 3pt]
			(v') edge[->] (u')
			(x') edge[<-] (v')
			;
			\end{tikzpicture}
			
			\vspace{.9cm}
			\begin{tikzpicture}[main node/.style={circle, draw,thick, inner sep=1pt,minimum size=18pt}, scale=.74,node distance=1cm]
			\tikzstyle{every loop}=[looseness=14]
			
			\node[main node] (v) at (0,0) {$v$} ;
			\node[main node] (u) at (2,0) {$u$} ;
			\node[main node] (x) at (1, {2*cos(30)} ) {$x$} ;
			
			\draw[very thick,-,shorten >= 3pt,shorten <= 3pt]
			(v) edge[->] (u)
			(x) edge[->] (u)
			;
			
			\node[main node] (v') at (4,0) {$v\ghostprime$} ;
			\node[main node] (u') at (6,0) {$u\ghostprime$} ;
			\node[main node] (x') at (5, {2*cos(30)} ) {$x\ghostprime$} ;
			\node[] (text') at (3,-.75) {$\mathcal P(U_{(x)},\phi)  = 0  \quad |\mathcal X(U_{(x)},\phi)| = 2$};
			\node[] (text') at (3,-1.25) {$\mathcal P(U_{(x)},\psi)  = 2 \quad |\mathcal X(U_{(x)},\psi)| = 0$};
			
			\draw[very thick,-,shorten >= 3pt,shorten <= 3pt]
			(v') edge[->] (u')
			(x') edge[<-] (v')
			(x') edge[->] (u')
			;
			\end{tikzpicture}
			
			\vspace{.9cm}
			\begin{tikzpicture}[main node/.style={circle, draw,thick, inner sep=1pt,minimum size=18pt}, scale=.74,node distance=1cm]
			\tikzstyle{every loop}=[looseness=14]
			
			\node[main node] (v) at (0,0) {$v$} ;
			\node[main node] (u) at (2,0) {$u$} ;
			\node[main node] (x) at (1, {2*cos(30)} ) {$x$} ;
			
			\draw[very thick,-,shorten >= 3pt,shorten <= 3pt]
			(v) edge[->] (u)
			(x) edge[<-] (v)
			;
			
			\node[main node] (v') at (4,0) {$v\ghostprime$} ;
			\node[main node] (u') at (6,0) {$u\ghostprime$} ;
			\node[main node] (x') at (5, {2*cos(30)} ) {$x\ghostprime$} ;
			\node[] (text') at (3,-.75) {$\mathcal P(U_{(x)},\phi)  = 0  \quad |\mathcal X(U_{(x)},\phi)| = 1$};
			\node[] (text') at (3,-1.25) {$\mathcal P(U_{(x)},\psi)  = 2 \quad |\mathcal X(U_{(x)},\psi)| = 0$};
			
			\draw[very thick,-,shorten >= 3pt,shorten <= 3pt]
			(v') edge[->] (u')
			(x') edge[<-] (v')
			;
			\end{tikzpicture}
			
			\vspace{.9cm}
			\begin{tikzpicture}[main node/.style={circle, draw,thick, inner sep=1pt,minimum size=18pt}, scale=.74,node distance=1cm]
			\tikzstyle{every loop}=[looseness=14]
			
			\node[main node] (v) at (0,0) {$v$} ;
			\node[main node] (u) at (2,0) {$u$} ;
			\node[main node] (x) at (1, {2*cos(30)} ) {$x$} ;
			
			\draw[very thick,-,shorten >= 3pt,shorten <= 3pt]
			(v) edge[->] (u)
			(x) edge[<-] (v)
			;
			
			\node[main node] (v') at (4,0) {$v\ghostprime$} ;
			\node[main node] (u') at (6,0) {$u\ghostprime$} ;
			\node[main node] (x') at (5, {2*cos(30)} ) {$x\ghostprime$} ;
			\node[] (text') at (3,-.75) {$\mathcal P(U_{(x)},\phi)  = 0  \quad |\mathcal X(U_{(x)},\phi)| = 2$};
			\node[] (text') at (3,-1.25) {$\mathcal P(U_{(x)},\psi)  =2 \quad |\mathcal X(U_{(x)},\psi)| = 0$};
			
			\draw[very thick,-,shorten >= 3pt,shorten <= 3pt]
			(v') edge[->] (u')
			(x') edge[<-] (v')
			(x') edge[->] (u')
			;
			\end{tikzpicture}
			
			\vspace{.9cm}
			\begin{tikzpicture}[main node/.style={circle, draw,thick, inner sep=1pt,minimum size=18pt}, scale=.74,node distance=1cm]
			\tikzstyle{every loop}=[looseness=14]
			
			\node[main node] (v) at (0,0) {$v$} ;
			\node[main node] (u) at (2,0) {$u$} ;
			\node[main node] (x) at (1, {2*cos(30)} ) {$x$} ;
			
			\draw[very thick,-,shorten >= 3pt,shorten <= 3pt]
			(v) edge[->] (u)
			(x) edge[<-] (v)
			(x) edge[->] (u)
			;
			
			\node[main node] (v') at (4,0) {$v\ghostprime$} ;
			\node[main node] (u') at (6,0) {$u\ghostprime$} ;
			\node[main node] (x') at (5, {2*cos(30)} ) {$x\ghostprime$} ;
			\node[] (text') at (3,-.75) {$\mathcal P(U_{(x)},\phi)  = 0  \quad |\mathcal X(U_{(x)},\phi)| = 3$};
			\node[] (text') at (3,-1.25) {$\mathcal P(U_{(x)},\psi)  = 3 \quad |\mathcal X(U_{(x)},\psi)| = 0$};
			
			\draw[very thick,-,shorten >= 3pt,shorten <= 3pt]
			(v') edge[->] (u')
			(x') edge[<-] (v')
			(x') edge[->] (u')
			;
			\end{tikzpicture}
			
			\vspace{.9cm}
			\begin{tikzpicture}[main node/.style={circle, draw,thick, inner sep=1pt,minimum size=18pt}, scale=.74,node distance=1cm]
			\tikzstyle{every loop}=[looseness=14]
			
			\node[main node] (v) at (0,0) {$v$} ;
			\node[main node] (u) at (2,0) {$u$} ;
			\node[main node] (x) at (1, {2*cos(30)} ) {$x$} ;
			
			\draw[very thick,-,shorten >= 3pt,shorten <= 3pt]
			(v) edge[->] (u)
			;
			
			\node[main node] (v') at (4,0) {$v\ghostprime$} ;
			\node[main node] (u') at (6,0) {$u\ghostprime$} ;
			\node[main node] (x') at (5, {2*cos(30)} ) {$x\ghostprime$} ;
			\node[] (text') at (3,-.75) {$\mathcal P(W,\phi)  = 0 \quad |\mathcal X(W,\phi)| = 1$};
			\node[] (text') at (3,-1.25) {$\mathcal P(W,\psi)  = 1\quad |\mathcal X(W,\psi)| = 0$};
			
			\draw[very thick,-,shorten >= 3pt,shorten <= 3pt]
			(v') edge[->] (u')
			(x') edge[->] (v')
			(x') edge[->] (u')
			;
			\end{tikzpicture}
			
			\vspace{.9cm}
			\begin{tikzpicture}[main node/.style={circle, draw,thick, inner sep=1pt,minimum size=18pt}, scale=.74,node distance=1cm]
			\tikzstyle{every loop}=[looseness=14]
			
			\node[main node] (v) at (0,0) {$v$} ;
			\node[main node] (u) at (2,0) {$u$} ;
			\node[main node] (x) at (1, {2*cos(30)} ) {$x$} ;
			
			\draw[very thick,-,shorten >= 3pt,shorten <= 3pt]
			(v) edge[->] (u)
			;
			
			\node[main node] (v') at (4,0) {$v\ghostprime$} ;
			\node[main node] (u') at (6,0) {$u\ghostprime$} ;
			\node[main node] (x') at (5, {2*cos(30)} ) {$x\ghostprime$} ;
			\node[] (text') at (3,-.75) {$\mathcal P(W,\phi)  = 0 \quad |\mathcal X(W,\phi)| = 1$};
			\node[] (text') at (3,-1.25) {$\mathcal P(W,\psi)  = 1\quad |\mathcal X(W,\psi)| = 0$};
			
			\draw[very thick,-,shorten >= 3pt,shorten <= 3pt]
			(v') edge[->] (u')
			(x') edge[->] (u')
			;
			\end{tikzpicture}
			
			\vspace{.9cm}
			\begin{tikzpicture}[main node/.style={circle, draw,thick, inner sep=1pt,minimum size=18pt}, scale=.74,node distance=1cm]
			\tikzstyle{every loop}=[looseness=14]
			
			\node[main node] (v) at (0,0) {$v$} ;
			\node[main node] (u) at (2,0) {$u$} ;
			\node[main node] (x) at (1, {2*cos(30)} ) {$x$} ;
			
			\draw[very thick,-,shorten >= 3pt,shorten <= 3pt]
			(v) edge[->] (u)
			;
			
			\node[main node] (v') at (4,0) {$v\ghostprime$} ;
			\node[main node] (u') at (6,0) {$u\ghostprime$} ;
			\node[main node] (x') at (5, {2*cos(30)} ) {$x\ghostprime$} ;
			\node[] (text') at (3,-.75) {$\mathcal P(W,\phi)  = 0 \quad |\mathcal X(W,\phi)| = 1$};
			\node[] (text') at (3,-1.25) {$\mathcal P(W,\psi)  = 1\quad |\mathcal X(W,\psi)| = 0$};
			
			\draw[very thick,-,shorten >= 3pt,shorten <= 3pt]
			(v') edge[->] (u')
			(x') edge[<-] (v')
			
			;
			\end{tikzpicture}
			
			\vspace{.9cm}
			\begin{tikzpicture}[main node/.style={circle, draw,thick, inner sep=1pt,minimum size=18pt}, scale=.74,node distance=1cm]
			\tikzstyle{every loop}=[looseness=14]
			
			\node[main node] (v) at (0,0) {$v$} ;
			\node[main node] (u) at (2,0) {$u$} ;
			\node[main node] (x) at (1, {2*cos(30)} ) {$x$} ;
			
			\draw[very thick,-,shorten >= 3pt,shorten <= 3pt]
			(v) edge[->] (u)
			;
			
			\node[main node] (v') at (4,0) {$v\ghostprime$} ;
			\node[main node] (u') at (6,0) {$u\ghostprime$} ;
			\node[main node] (x') at (5, {2*cos(30)} ) {$x\ghostprime$} ;
			\node[] (text') at (3,-.75) {$\mathcal P(W,\phi)  = 0 \quad |\mathcal X(W,\phi)| = 1$};
			\node[] (text') at (3,-1.25) {$\mathcal P(W,\psi)  = 1\quad |\mathcal X(W,\psi)| = 0$};
			
			\draw[very thick,-,shorten >= 3pt,shorten <= 3pt]
			(v') edge[->] (u')
			(x') edge[<-] (v')
			(x') edge[->] (u')
			;
			\end{tikzpicture}
			
			\vspace{.9cm}
			\begin{tikzpicture}[main node/.style={circle, draw,thick, inner sep=1pt,minimum size=18pt}, scale=.74,node distance=1cm]
			\tikzstyle{every loop}=[looseness=14]
			
			\node[main node] (v) at (0,0) {$v$} ;
			\node[main node] (u) at (2,0) {$u$} ;
			\node[main node] (x) at (1, {2*cos(30)} ) {$x$} ;
			
			\draw[very thick,-,shorten >= 3pt,shorten <= 3pt]
			(v) edge[->] (u)
			;
			
			\node[main node] (v') at (4,0) {$v\ghostprime$} ;
			\node[main node] (u') at (6,0) {$u\ghostprime$} ;
			\node[main node] (x') at (5, {2*cos(30)} ) {$x\ghostprime$} ;
			\node[] (text') at (3,-.75) {$\mathcal P(W,\phi)  = 0 \quad |\mathcal X(W,\phi)| = 1$};
			\node[] (text') at (3,-1.25) {$\mathcal P(W,\psi)  = 1\quad |\mathcal X(W,\psi)| = 0$};
			
			\draw[very thick,-,shorten >= 3pt,shorten <= 3pt]
			(v') edge[->] (u')
			(x') edge[<-] (v')
			(x') edge[<-] (u')
			;
			\end{tikzpicture}
			
			\vspace{.9cm}
			\begin{tikzpicture}[main node/.style={circle, draw,thick, inner sep=1pt,minimum size=18pt}, scale=.74,node distance=1cm]
			\tikzstyle{every loop}=[looseness=14]
			
			\node[main node] (v) at (0,0) {$v$} ;
			\node[main node] (u) at (2,0) {$u$} ;
			\node[main node] (x) at (1, {2*cos(30)} ) {$x$} ;
			
			\draw[very thick,-,shorten >= 3pt,shorten <= 3pt]
			(v) edge[->] (u)
			;
			
			\node[main node] (v') at (4,0) {$v\ghostprime$} ;
			\node[main node] (u') at (6,0) {$u\ghostprime$} ;
			\node[main node] (x') at (5, {2*cos(30)} ) {$x\ghostprime$} ;
			\node[] (text') at (3,-.75) {$\mathcal P(W,\phi)  = 0 \quad |\mathcal X(W,\phi)| = 1$};
			\node[] (text') at (3,-1.25) {$\mathcal P(W,\psi)  = 1\quad |\mathcal X(W,\psi)| = 0$};
			
			\draw[very thick,-,shorten >= 3pt,shorten <= 3pt]
			(v') edge[->] (u')
			(x) edge[->] (v)
			(x) edge[->] (u)
			;
			\end{tikzpicture}
			
			\vspace{.9cm}
			\begin{tikzpicture}[main node/.style={circle, draw,thick, inner sep=1pt,minimum size=18pt}, scale=.74,node distance=1cm]
			\tikzstyle{every loop}=[looseness=14]
			
			\node[main node] (v) at (0,0) {$v$} ;
			\node[main node] (u) at (2,0) {$u$} ;
			\node[main node] (x) at (1, {2*cos(30)} ) {$x$} ;
			
			\draw[very thick,-,shorten >= 3pt,shorten <= 3pt]
			(v) edge[->] (u)
			;
			
			\node[main node] (v') at (4,0) {$v\ghostprime$} ;
			\node[main node] (u') at (6,0) {$u\ghostprime$} ;
			\node[main node] (x') at (5, {2*cos(30)} ) {$x\ghostprime$} ;
			\node[] (text') at (3,-.75) {$\mathcal P(W,\phi)  = 0 \quad |\mathcal X(W,\phi)| = 1$};
			\node[] (text') at (3,-1.25) {$\mathcal P(W,\psi)  = 1\quad |\mathcal X(W,\psi)| = 0$};
			
			\draw[very thick,-,shorten >= 3pt,shorten <= 3pt]
			(v') edge[->] (u')
			(x) edge[->] (u)
			;
			\end{tikzpicture}
			
			\vspace{.9cm}
			\begin{tikzpicture}[main node/.style={circle, draw,thick, inner sep=1pt,minimum size=18pt}, scale=.74,node distance=1cm]
			\tikzstyle{every loop}=[looseness=14]
			
			\node[main node] (v) at (0,0) {$v$} ;
			\node[main node] (u) at (2,0) {$u$} ;
			\node[main node] (x) at (1, {2*cos(30)} ) {$x$} ;
			
			\draw[very thick,-,shorten >= 3pt,shorten <= 3pt]
			(v) edge[->] (u)
			;
			
			\node[main node] (v') at (4,0) {$v\ghostprime$} ;
			\node[main node] (u') at (6,0) {$u\ghostprime$} ;
			\node[main node] (x') at (5, {2*cos(30)} ) {$x\ghostprime$} ;
			\node[] (text') at (3,-.75) {$\mathcal P(W,\phi)  = 0 \quad |\mathcal X(W,\phi)| = 1$};
			\node[] (text') at (3,-1.25) {$\mathcal P(W,\psi)  = 1\quad |\mathcal X(W,\psi)| = 0$};
			
			\draw[very thick,-,shorten >= 3pt,shorten <= 3pt]
			(v') edge[->] (u')
			(x) edge[<-] (v)
			
			;
			\end{tikzpicture}
			
			\vspace{.9cm}
			\begin{tikzpicture}[main node/.style={circle, draw,thick, inner sep=1pt,minimum size=18pt}, scale=.74,node distance=1cm]
			\tikzstyle{every loop}=[looseness=14]
			
			\node[main node] (v) at (0,0) {$v$} ;
			\node[main node] (u) at (2,0) {$u$} ;
			\node[main node] (x) at (1, {2*cos(30)} ) {$x$} ;
			
			\draw[very thick,-,shorten >= 3pt,shorten <= 3pt]
			(v) edge[->] (u)
			;
			
			\node[main node] (v') at (4,0) {$v\ghostprime$} ;
			\node[main node] (u') at (6,0) {$u\ghostprime$} ;
			\node[main node] (x') at (5, {2*cos(30)} ) {$x\ghostprime$} ;
			\node[] (text') at (3,-.75) {$\mathcal P(W,\phi)  = 0 \quad |\mathcal X(W,\phi)| = 1$};
			\node[] (text') at (3,-1.25) {$\mathcal P(W,\psi)  = 1\quad |\mathcal X(W,\psi)| = 0$};
			
			\draw[very thick,-,shorten >= 3pt,shorten <= 3pt]
			(v') edge[->] (u')
			(x) edge[<-] (v)
			(x) edge[->] (u)
			;
			\end{tikzpicture}
			
			\vspace{.9cm}
			\begin{tikzpicture}[main node/.style={circle, draw,thick, inner sep=1pt,minimum size=18pt}, scale=.74,node distance=1cm]
			\tikzstyle{every loop}=[looseness=14]
			
			\node[main node] (v) at (0,0) {$v$} ;
			\node[main node] (u) at (2,0) {$u$} ;
			\node[main node] (x) at (1, {2*cos(30)} ) {$x$} ;
			
			\draw[very thick,-,shorten >= 3pt,shorten <= 3pt]
			(v) edge[->] (u)
			;
			
			\node[main node] (v') at (4,0) {$v\ghostprime$} ;
			\node[main node] (u') at (6,0) {$u\ghostprime$} ;
			\node[main node] (x') at (5, {2*cos(30)} ) {$x\ghostprime$} ;
			\node[] (text') at (3,-.75) {$\mathcal P(W,\phi)  = 0 \quad |\mathcal X(W,\phi)| = 1$};
			\node[] (text') at (3,-1.25) {$\mathcal P(W,\psi)  = 1\quad |\mathcal X(W,\psi)| = 0$};
			
			\draw[very thick,-,shorten >= 3pt,shorten <= 3pt]
			(v') edge[->] (u')
			(x) edge[<-] (v)
			(x) edge[<-] (u)
			;
			\end{tikzpicture}
			
			\vspace{.9cm}
			\begin{tikzpicture}[main node/.style={circle, draw,thick, inner sep=1pt,minimum size=18pt}, scale=.74,node distance=1cm]
			\tikzstyle{every loop}=[looseness=14]
			
			\node[main node] (v) at (0,0) {$v$} ;
			\node[main node] (u) at (2,0) {$u$} ;
			\node[main node] (x) at (1, {2*cos(30)} ) {$x$} ;
			
			\draw[very thick,-,shorten >= 3pt,shorten <= 3pt]
			(v) edge[->] (u)
			;
			
			\node[main node] (v') at (4,0) {$v\ghostprime$} ;
			\node[main node] (u') at (6,0) {$u\ghostprime$} ;
			\node[main node] (x') at (5, {2*cos(30)} ) {$x\ghostprime$} ;
			\node[] (text') at (3,-.75) {$\mathcal P(W,\phi)  = 0 \quad |\mathcal X(W,\phi)| = 1$};
			\node[] (text') at (3,-1.25) {$\mathcal P(W,\psi)  = 1\quad |\mathcal X(W,\psi)| = 0$};
			
			\draw[very thick,-,shorten >= 3pt,shorten <= 3pt]
			(v') edge[->] (u')
			;
			\end{tikzpicture}
			\vspace{.9cm}

		\end{center}
		
	\end{multicols}
	
\section{Algorithms}\label{sec:alg}
The following pseudocode, written in Python style, gives an algorithm for calculating DMCES($G,G'$) for node-labeled digraphs $G = (V, \mathcal D, \ell_v,\emptyset)$ and $G' = (V', \mathcal D', \ell'_v,\emptyset)$ that leverages Theorem~\ref{thr:size}.
The idea of the algorithm is to, at every recursive call, create a separate branch for each way we can grow the current feasible solution $(U,\phi)$ into a feasible solution $(U', \phi')$ such that $U \subseteq U'$ and $\phi'|_U = \phi$ (see Figure~\ref{fig:recursion}).
We do not keep track of $U$ explicitly, rather it is the domain of the map $\phi$.  The map $\phi$ is represented as a set of ordered pairs $\Phi \subseteq V \times V'$ with the property that if  $(v_1,v_1') , (v_2,v_2')  \in V \times V'$, then $v_1 \neq v_2$ and $v'_1 \neq v'_2$. Moreover, $\ell_v(v_i) = \ell'_v(v_i')$ whenever $(v_i,v_i') \in \Phi$. The pairs $\Phi$ define a label-preserving, injective function $\phi: U \to V'$. As a shortcut, we sometimes refer to the collection of first elements of $\Phi$ as the domain of $\Phi$, and similarly refer to the second elements as the range or image. These are in fact the domain and range of $\phi$. 
The set $X \subseteq V$  is an  ordered list of nodes containing all nodes as yet unassigned to $\Phi$.

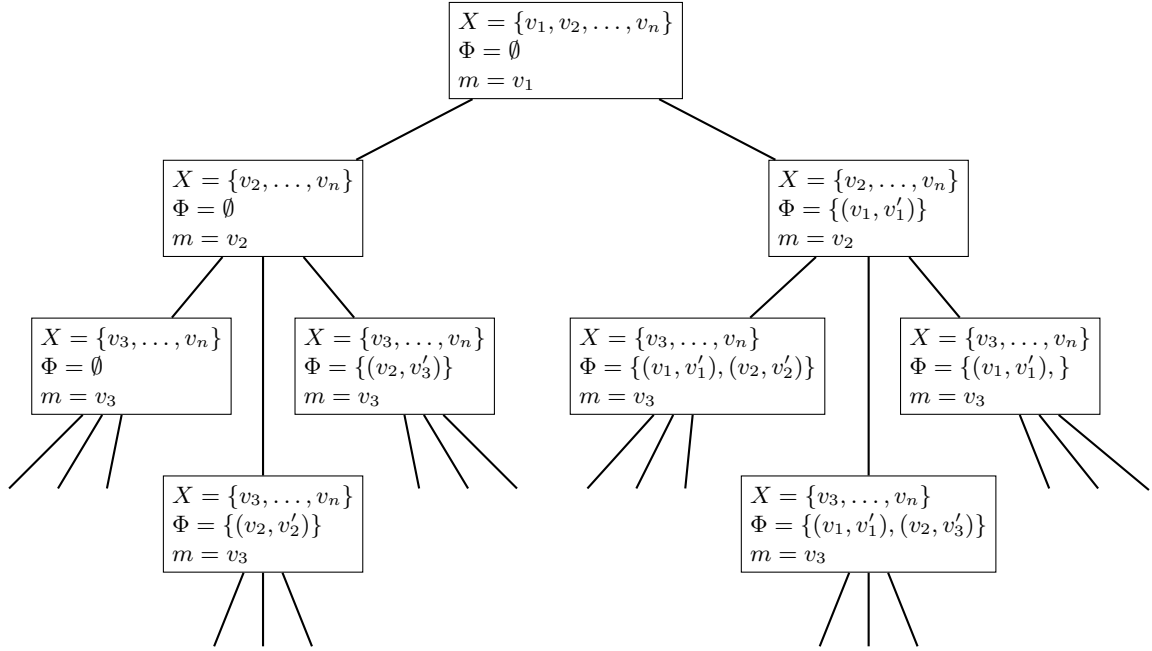
\begin{figure}[h!]
	\begin{center}
		\begin{tikzpicture}[main node/.style={rectangle,fill=white!20,draw,font=\sffamily\small},scale=1.4, align = left]
		\node[main node] (1) at (.375,0) {$X = \{v_1,v_2,\dots,v_n\}
			$\\$ \Phi= \emptyset
			$\\$m=v_1	$};
		\node[main node] (2) at (-2.5,-1.5) {$X = \{v_2,\dots,v_n\}
			$\\$ \Phi= \emptyset
			$\\$m=v_2	$};
		\node[main node] (3) at (3.25,-1.5) {$X = \{v_2,\dots,v_n\}
			$\\$ \Phi= \{(v_1,v'_1)\}
			$\\$m=v_2	$};
		\node[main node] (5) at (-3.75,-3) {$X = \{v_3,\dots,v_n\}
			$\\$ \Phi= \emptyset
			$\\$m=v_3	$};
		\node[main node] (6) at (-2.5,-4.5) {$X = \{v_3,\dots,v_n\}
			$\\$ \Phi= \{(v_2,v'_2)\}
			$\\$m=v_3	$};
		\node[main node] (7) at (-1.25,-3) {$X = \{v_3,\dots,v_n\}
			$\\$ \Phi= \{(v_2,v'_3)\}
			$\\$m=v_3	$};
		\node[main node] (8) at (4.5,-3) {$X = \{v_3,\dots,v_n\}
			$\\$ \Phi= \{(v_1,v'_1),\}
			$\\$m=v_3	$};
		\node[main node] (9) at (3.25,-4.5) {$X = \{v_3,\dots,v_n\}
			$\\$ \Phi= \{(v_1,v'_1),(v_2,v'_3)\}
			$\\$m=v_3	$};
		\node[main node] (10) at (1.6255,-3) {$X = \{v_3,\dots,v_n\}
			$\\$ \Phi= \{(v_1,v'_1),(v_2,v'_2)\}
			$\\$m=v_3	$};
		\node[] (5a) at (-4.5,-4.25) {};
		\node[] (5b) at (-5,-4.25) {};
		\node[] (5c) at (-4,-4.25) {};
		\node[] (7a) at (4-4.5,-4.25) {};
		\node[] (7b) at (4-5,-4.25) {};
		\node[] (7c) at (4-4,-4.25) {};
		\node[] (10a) at (5.5-4.5,-4.25) {};
		\node[] (10b) at (5.5-5,-4.25) {};
		\node[] (10c) at (5.5-4,-4.25) {};
		\node[] (8a) at (10-4.5,-4.25) {};
		\node[] (8b) at (10-5,-4.25) {};
		\node[] (8c) at (10-4,-4.25) {};
		\node[] (6a) at (-2.5,-5.75) {};
		\node[] (6b) at (-3,-5.75) {};
		\node[] (6c) at (-2,-5.75) {};
		\node[] (9a) at (3.25,-5.75) {};
		\node[] (9b) at (3.75,-5.75) {};
		\node[] (9c) at (2.75,-5.75) {};
		
		\path[->,>=angle 90,thick]
		(1) edge[-] node[] {} (2)
		(1) edge[-] node[] {} (3)
		(2) edge[-] node[] {} (5)
		(2) edge[-] node[] {} (6)
		(2) edge[-] node[] {} (7)
		(3) edge[-] node[] {} (8)
		(3) edge[-] node[] {} (9)
		(3) edge[-] node[] {} (10)
		(5) edge[-] node[] {} (5a)
		(5) edge[-] node[] {} (5b)
		(5) edge[-] node[] {} (5c)
		(7) edge[-] node[] {} (7a)
		(7) edge[-] node[] {} (7b)
		(7) edge[-] node[] {} (7c)
		(8) edge[-] node[] {} (8a)
		(8) edge[-] node[] {} (8b)
		(8) edge[-] node[] {} (8c)
		(10) edge[-] node[] {} (10a)
		(10) edge[-] node[] {} (10b)
		(10) edge[-] node[] {} (10c)
		(6) edge[-] node[] {} (6a)
		(6) edge[-] node[] {} (6b)
		(6) edge[-] node[] {} (6c)
		(9) edge[-] node[] {} (9a)
		(9) edge[-] node[] {} (9b)
		(9) edge[-] node[] {} (9c)
		;
		
		\end{tikzpicture}
	\end{center}
	\caption{The head of the $\tt pick\_nodes()$ recursion tree. Each box is an instance of the function, where $X$ and $\Phi$ are the input parameters. $m$ is found by simply taking the first element of $X$. Lines indicate which function makes each recursive call. In this example $\ell_v(v_1) = \ell_v'(v'_1)$ and $\ell_v(v_2) = \ell'_v(v'_2) = \ell'_v(v'_3)$.}
	\label{fig:recursion}
\end{figure}

At each recursive call of \verb|pick_nodes()| the function parameters are the list of nodes $X \subseteq V$ and a set of ordered pairs of nodes $ \Phi  \subseteq V \times V'$, see Figure~\ref{fig:recursion}.  At the initial call of $\verb|pick_nodes()|$, $X=V$ and $\Phi = \emptyset$. The first element of the list $X$, $X$[0], is stored as $m$. The function then determines all possible nodes $n \in V'$ that both share a label with $m$ and do not appear in any element of $\Phi$. For each such $m$ a new recursive call $\verb|pick_nodes|(X', \Phi ')$ is made in which $X' = X \setminus  \{m\}$ and $\Phi' = \Phi \cup \{(m,n)\}$. A new recursive call may also be made with $X' = X\setminus \{m\}$ and  $\Phi' =\Phi$, if adding the edge $\{(m,n)\}$ to $\Phi$ would exceed the maximum node count. This is checked by the line: if $|\Phi_L| +|X_L| > \verb|final_num_nodes|[L]$,  where for any set $Z$, we write $Z_L$ to indicate the subset of $Z$ which contains all elements of $Z$ with label $L$. Note that  $\verb|final_num_nodes|[L]= N_L(G,G')$, as calculated in the proof of Theorem~\ref{thr:size}.
This recursion continues until an instance occurs with $X = \emptyset$ at which point the score of $\Phi$, given by Equation~\eqref{eq:score}, is calculated and returned.

During each instance of $\verb|pick_nodes()|$ the return values of all recursive calls made within the instance are compared and the largest is returned. In this way, only the value from the branch of the recursive tree that corresponds to the largest maximal solution will be returned all the way to the top of the tree. If a branch is not viable, that is $X \neq \emptyset$ and no more recursive calls can be made, then $0$ is returned. 

Since each node $v \in V$ corresponds to a different level of recursion, and since a recursive call is made for all possible pairings of $v$ to a node in $V'$, there will be a branch for every maximal cardinality solution. Therefore, since the graph size resulting from every branch is compared, the maximum common subgraph size $\DMCES(G,G')$ will be returned. 
\\\\
\textbf{Algorithm 1.}
\\\\
\noindent def \verb|DMCES|($G,G'$)\\
\indent global \verb|final_num_nodes| = \verb|find_final_num_nodes|($G,G'$)\\
\indent return \verb|pick_nodes| ($\text{nodes}(G),\emptyset))$\\\\\noindent def \verb|pick_nodes|($X,\Phi$)\\
\indent if $X == \emptyset$ \\
\indent \indent	return $\mathcal P(\Phi)$\\
\indent \verb|score| = 0\\
\indent	$m = X$[0]\\
\indent $L= \ell_v(m)$\\
\indent for $n \in V'$ such that $\ell'_v(n) == L$ and $n$ not in the image of $\Phi$\\
\indent \indent \verb|score|= $\max( \verb|score|,\verb|pick_nodes|(X \setminus \{m\}, \Phi\cup \{(m,n)\}))$\\
\indent	if $|\Phi_L| +|X_L| > \verb|final_num_nodes|[L]$\\
\indent \indent	\verb|score| = $\max(\verb|score|, \verb|pick_nodes|(X \setminus\{m\}, \Phi))$
\\

Suppose now that $G$ and $G'$ are transitive closures. Using Corollary~\ref{thr:direction} we can improve our algorithm to only consider solutions that are order-respecting. The predecessors of a node $v$ are all nodes $u$ such that there is a path from $u$ to $v$.
\\\\
\textbf{Algorithm 2.}
\\\\
\noindent def \verb|DMCES|($G,G'$)\\
\indent global \verb|final_num_nodes| = \verb|find_final_num_nodes|($G,G'$)\\
\indent return \verb|pick_nodes| (\verb|topologically_sort|$(\text{nodes}(G)), \emptyset))$\\\\\noindent def \verb|pick_nodes|($X,\Phi$)\\
\indent if $X == \emptyset$ \\
\indent \indent	return $\mathcal P(\Phi)$\\
\indent \verb|score| = 0\\
\indent	$m = X$[0]\\
\indent $L= \ell_v(m)$\\
\indent \verb|cross| = $\text{predecessors}(\phi(\{ n \in \text{predecessors}(k)\;|\; k \in \text{domain of $\Phi$}\} ))$\\
\indent for $n \in V'$ such that $\ell'_v(n) == L$,  $n$ not in image of $\Phi$, and $n$ not in \verb|cross|\\
\indent \indent	\verb|score|= $\max( \verb|score|,\verb|pick_nodes|(X \setminus \{m\}, \Phi\cup \{(m,n)\}))$\\
\indent	if $|\Phi_L| +|X_L| > \verb|final_num_nodes|[L]$\\
\indent \indent	\verb|score| = $\max(\verb|score|, \verb|pick_nodes|(X \setminus\{m\}, \Phi))$
\\\\
 Here, \verb|cross| is the set of nodes $n$ for which adding $(m,n)$ to $\Phi$ would cause $\Phi$ not to respect order on labels. 
 
We can further improve the algorithm in the case when  subgraphs induced by all nodes of a given label are directed path graphs.
That is, graphs induced by $U \subseteq V$, where $\ell_v(u) = \ell_v(v)$ for all $u,v \in U$, are isomorphic to graphs of form $(U, \bar{\mathcal D}, \ell_v|_U,\emptyset)$, where
$$U = \{v_1,v_2,...,v_n\}$$
$$\bar{\mathcal D} = \{(v_1,v_2), (v_2,v_3),...,(v_{n-1}, v_n)\}.$$
Digraphs of partial orders produced from time series using the technique in~\cite{tspo} always satisfy this property.
With this added assumption we can further improve the algorithm by  keeping track of what nodes in $V'$ will not be in the image of $\phi'$ for any extension $(U',\phi')$ of a feasible solution $(U,\phi)$. This is stored as the set $Y \subseteq V'$ in the following algorithm. If adding $(m,n)$ to $\Phi$ will cause $|V_L'| - |Y_L| <\verb|final_num_nodes|[L]$ then the branch is not continued as it can not lead to a maximal cardinality solution. 
\\\\
\textbf{Algorithm 3.}
\\\\
\noindent def \verb|DMCES|($G,G'$)\\
\indent global \verb|final_num_nodes| = \verb|find_final_num_nodes|($G,G'$)\\
\indent return \verb|pick_nodes| (\verb|topologically_sort|$(\text{nodes}(G)),\emptyset, \emptyset))$\\\\
\noindent def \verb|pick_nodes|($X,Y, \Phi$)\\
\indent if $X$ is empty\\
\indent \indent	return $\mathcal P(\Phi)$\\
\indent \verb|score| = 0\\
\indent	$m = X$[0]\\
\indent $L= \ell_v(m)$\\
\indent for $n \in V'$ such that $\ell'_v(n) == L$, $n \notin Y$\\
\indent \indent	$\hat Y$= \{$v \in \verb|predeccessors|(n)|\text{  $v$ not in range of $\Phi$ and $\ell'_v(v) == L$}$\}\\
\indent \indent \indent	if \verb|final_num_nodes|[L]$\leq |V'_L|-|Y_L \cup \hat Y _L|$ \\
\indent \indent \indent \indent	\verb|score|= $\max( \verb|score|,\verb|pick_nodes|(X \setminus \{m\},Y \cup \hat Y, \Phi\cup \{(m,n)\}))$\\
\indent	if $|\Phi_L| +|X_L| > \verb|final_num_nodes|[L]$\\
\indent \indent	\verb|score| = $\max(\verb|score|, \verb|pick_nodes|(X \setminus\{m\},Y, \Phi))$
\\

%

\end{document}